\documentclass[11pt, letterpaper,reqno]{article}

\usepackage{amsmath,amssymb,amsthm, amsfonts, amstext}
\usepackage[numbers]{natbib}
\usepackage{graphicx}
\usepackage{microtype}
\usepackage{hyperref}
\hypersetup{
  colorlinks=true,
  linkcolor=blue,
  citecolor=blue,
  urlcolor=black
}

\usepackage{subcaption}
\usepackage{chngcntr} % counter reset in figures, etc
\usepackage{authblk}
\usepackage{color}
\usepackage[pagewise]{lineno}
\usepackage{setspace}
\newtheorem{Proposition}{Proposition}[section]
\newtheorem{Corollary}{Corollary}[section]

\numberwithin{equation}{section}
\counterwithout{equation}{section}

\newcommand{\R}{\mathbb{R}}

\def \F{\mathbb{F}}

\def \P{\mathbb{P}}

\title{From debt crises to financial crashes (and back): a stock--flow consistent model for stock price bubbles}
\author[1]{Matheus R. Grasselli}
\author[2]{Adrien Nguyen Huu}
\affil[1]{McMaster University, Hamilton, Canada}
\affil[2]{CEE-M, Univ. Montpellier, CNRS, Institut Agro, Montpellier, France}
\date{\today}

\begin{document}
% \title[Debt crises and financial crashes]{From debt crises to financial crashes (and back): a stock--flow consistent model for stock price bubbles}

% \author[M. R. Grasselli]{Matheus R. Grasselli}
% \address{McMaster University, Hamilton, Canada}
% \email{Corresponding author: grasselli@mcmaster.ca}

% \author[A. Nguyen-Huu]{Adrien Nguyen-Huu}
% \address{CEE-M, Univ. Montpellier, CNRS, Institut Agro, Montpellier, France}
% \email{adrien.nguyen-huu@umontpellier.fr}

% \title[Debt crises and financial crashes]{From debt crises to financial crashes (and back): a stock-flow consistent model for stock price bubbles}

% \author{Matheus R. Grasselli$^{1}$ \And Adrien Nguyen-Huu$^{2}$}

% \thanks{$^{1}$McMaster University, Hamilton, Canada. Corresponding Author: \texttt{grasselli@mcmaster.ca}. }
% \thanks{$^{2}$CEE-M, Univ. Montpellier, CNRS, Institut Agro, Montpellier, France. email: \texttt{adrien.nguyen-huu@umontpellier.fr}}

%%%%%%%%%%%%%%%%%%%%%%%%%%%%%%%%%%%%%%%%%%%%%%%%%%%%%%%%%%%%%%%%%%%%%%%%%%%
%%%%%%%%%%%%%%%%%%%%%%%%%%%%%%%%%%%%%%%%%%%%%%%%%%%%%%%%%%%%%%%%%%%%%%%%%%%
%%%%%%%%%%%%%%%%%%%%%%%%%%%%%%%%%%%%%%%%%%%%%%%%%%%%%%%%%%%%%%%%%%%%%%%%%%%

\maketitle

% \begin{abstract}
% [REWRITE] We merge two models by Steve Keen, namely a monetary model of debt-deflation
% and a version with Ponzi destabilization. We recall and study old and new equilibrium properties,
% with local stability analysis.
% We then add an auxiliary stochastic model of financial markets based on a jump--diffusion process
% with endogenous jump intensity.
% This model captures main characteristics of Hyman Minsky's \emph{Financial Instability Hypothesis} (FIH),
% and the \emph{Quantitative Theory of Credit} (QTC) of Richard Werner,
% with an asset price bubble fueled by pure speculative credit and Market crashes impacting the real economy.
% We develop and study fundamental and numerical properties of this model, and its suitability to explain
% Financial crisis and the relationship between growth and interest.
% \end{abstract}

\begin{abstract}
We develop a stochastic macro--financial model in continuous time by integrating two specifications of the Keen economic framework with a financial market driven by a jump-diffusion process. The economic block of the model combines monetary debt-deflation mechanisms with Ponzi-type financial destabilization and is influenced by the financial market through a stochastic interest rate that depends on asset price returns. The financial market block of the model consists of an asset with jump--diffusion price process with endogenous, state-dependent jump intensities driven by speculative credit flows.

The model formalizes a feedback loop linking credit expansion, crash risk, perceived return dynamics, and bank lending spreads. Under suitable parameter restrictions, we establish global existence and non-explosion of the coupled system.

Numerical experiments illustrate how variations in credit sensitivity and jump parameters generate regimes ranging from stable growth to recurrent boom--bust cycles. The framework provides a tractable setting for analyzing endogenous financial fragility within a mathematically well-posed macro--financial system.
\end{abstract}

%\subjclass[2020]{91B55, 60H10, 60G44, 34D20, 37N40}

%\keywords{Asset price bubbles; Jump--diffusion; Stock--flow consistency; Endogenous crash risk; Credit dynamics; Financial instability}

\section{Introduction}
\label{sec:intro}

The causal narrative of modern crises is often told as a transition \emph{from a financial crash to a debt crisis}, emphasizing how banking turmoil propagates into sovereign and private balance sheets \citep{ReinhartRogoff2011}. The present paper completes the cycle by emphasizing the reverse direction: how a debt-driven economy, in which credit creation is endogenous, can generate asset price booms and abrupt market dislocations that then feed back into the real sector. This perspective is aligned with the empirical and institutional view that most money is created by commercial banks when they extend credit, rather than being mechanically “multiplied” from reserves \citep{McLeayRadiaThomas2014,Werner2014}. It also resonates with Minsky’s Financial Instability Hypothesis, where tranquil periods endogenously breed fragility, and where leverage and refinancing conditions are central state variables \citep{Keen1995,Keen2013}.

In parallel, several strands of the bubbles literature in economics and finance have progressively moved beyond the textbook “rational bubble” paradigm. In discrete time and in equilibrium settings, rational expectations bubbles have been studied as self-fulfilling components consistent with no-arbitrage and market clearing \citep{BlanchardWatson1982,Tirole1985}. In continuous-time mathematical finance, an influential approach formalizes bubbles through \emph{strict local martingales} under a pricing measure: discounted asset prices remain local martingales but fail to be true martingales, with tangible implications for option pricing, parity relations, and maturity monotonicity \citep{CoxHobson2005,HestonLoewensteinWillard2007,JarrowProtterShimbo2010}. A complementary economics-based view connects bubble phenomena to trading constraints and market incompleteness \citep{LoewensteinWillard2000a,LoewensteinWillard2000b}, while a more recent probabilistic literature provides sharp characterizations of “bubble times” and their relationship with measure changes and martingale defects \citep{BiaginiFollmerNedelcu2014,KardarasKreherNikeghbali2015}.

A common limitation of most strict-local-martingale or rational-bubble specifications is that the drivers of booms and crashes are typically encoded in \emph{fixed} (or weakly time-varying) parameters: volatility, drift corrections, borrowing constraints, or jump features are postulated \emph{ex ante}, rather than being endogenously tied to credit creation, bank balance-sheet behavior, and debt servicing conditions. Conversely, macro-financial models that take endogenous credit seriously often rely on reduced-form or exogenous representations of financial markets. Stock--flow consistent (SFC) monetary macrodynamics, and in particular the line of models initiated by \citet{Keen1995}, provide a coherent accounting framework in which debt, income flows, and balance sheets evolve jointly, allowing one to study debt-deflation, leverage cycles, and financial fragility within a closed system of stocks and flows \citep{GrasselliCostaLima2012}. Several extensions have enriched this core by introducing inflation and speculative dynamics \citep{GrasselliNguyenHuu2015}, and by clarifying the role of debt-financed investment and instability channels \citep{NguyenHuuPottier2017}. 
Related empirical and microstructural contributions emphasize that market fluctuations can reflect balance-sheet capacity and demand elasticity for risky assets, not only changes in fundamentals \citep{bouchaud2022inelastic, gabaix2021search}.

The objective of this paper is to connect these two literatures on crises—SFC macrodynamics of credit economies and stochastic mathematical models of bubbles—within a single, tractable framework. Concretely, we merge two Keen-style economic blocks (a monetary debt-deflation mechanism and a Ponzi destabilization channel) with an auxiliary stochastic model for the stock market. The financial market is specified as a jump--diffusion process with \emph{endogenous jump intensity}, designed to capture a stylized Minskyan build-up of fragility and abrupt repricings. The economic block provides endogenous state variables (profits, leverage, debt ratios, and repayment capacity), while the market block provides a realistic representation of crashes as discontinuous events rather than smooth mean reversion. The bridge between the two is built to reflect the Quantitative Theory of Credit \citep{Werner2012,Werner2014}: speculative credit fuels asset price appreciation in expansions, while sudden devaluations interact with bank behavior and the cost of credit.

Our contribution is primarily methodological and exploratory. We propose a stock--flow consistent macro-financial model in which (i) the real economy subsystem retains the accounting discipline and instability mechanisms of the Keen tradition, (ii) the financial market subsystem belongs to the modern toolbox of stochastic calculus for bubbles (jump--diffusion dynamics with endogenous crash intensity), and (iii) the coupling between the two is economically interpretable and mathematically controlled. We derive and discuss key qualitative properties of the coupled system, with particular attention to the financial market component and its compatibility with no-arbitrage-inspired modeling principles. We then perform numerical simulations to obtain an intuitive, scenario-based understanding of the joint dynamics and to identify robust channels of amplification across regimes.

The main economic finding concerns the \emph{reaction channel} from financial turmoil to the credit economy. From the economy to markets, the model reproduces the well-known destabilizing role of speculation and credit expansion emphasized in SFC analyses \citep{GrasselliCostaLima2012}. From markets to the economy, the crucial mechanism is the banking sector’s effective lending rate: when markets become turbulent and crash risk rises, banks react in a risk-averse manner by increasing the credit premium, which shocks the interest burden and can destabilize debt-financed investment. The simulations highlight that the functional form and strength of this reaction function determine whether the economy can absorb market agitation or else transitions to a self-reinforcing cycle of contraction in credit, output, and profits.

The paper is organized as follows. 
In Section~\ref{sec:keen}, we develop a version of the Keen model \citep{Keen1995} with inflation and speculation along the lines introduced in \citep{GrasselliNguyenHuu2015}, and revisit its equilibrium and stability properties within a unified stock--flow consistent framework. 
In Section~\ref{sec:stochastic}, we introduce the stochastic model of financial markets motivated by the Quantitative Theory of Credit of \citet{Werner2012}, present the endogenous jump-intensity specification and the banking reaction mechanism, and state the main theoretical properties supporting these modeling choices. 
Section~\ref{sec:numerics} is devoted to numerical computations and simulations of the coupled system, with particular emphasis on how financial characteristics of the economy and banks’ reaction functions affect real-sector dynamics. 
Appendices gather detailed stability analysis, existence results, and numerical algorithms.

\section{The real economy}
\label{sec:keen}

\subsection{Accounting Framework}

We recall the setting in \cite{GrasselliNguyenHuu2015} for the 3-sector stock--flow consistent model corresponding to the balance sheets, transactions, and flow of funds shown in Table \ref{table}. In what follows, we denote a constant parameter by a bar above the letter corresponding to it, for example $\bar \delta$ for the constant depreciation rate, whereas any other quantities without a bar are understood to be time-varying. As usual, $\dot x$ denotes the time derivative of $x$.  

The balance sheets matrix shown in Table \ref{table} records the {\em stocks} (i.e., levels) of all assets and liabilities within the economy in units of currency (e.g., USD), with a positive entry representing an asset of the sector and a negative entry representing a liability. This matrix already embodies many of the assumptions in the model, notably: (i) households hold bank deposits $M_h$ as a monetary asset and also own all the firms and banks in the economy through private equity ${\mathcal E}_h$; (ii) firms hold all the productive capital in the economy, with nominal value $pK$, in addition to holding bank deposits $M_f$ and borrowing through bank loans $L$; and (iii) banks make loans and accept deposits. Accordingly, from the column sums in the balance sheet matrix, we obtain the following relationships between the equity and net worth of the sectors: 
\begin{align}
    X_h &=M_h +{\mathcal E}_h                  &&\text{(net worth of households)}\\
    {\mathcal E}_f &=pK +M_f -L           &&\text{(equity of firms)} \label{eqn:firms_net_worth} \\ 
    {\mathcal E}_b &= L - M,             &&\text{(equity of banks)}
\end{align} 
from which we obtain that the total net worth of the entire economy is equal to the nominal value of the total capital stock, as expected: 
\begin{equation}
 X = X_h + X_f + X_b = (M_h +{\mathcal E}_h)+(pK +M_f -L- {\mathcal E}_f) + (L - M - {\mathcal E}_b)= pK. 
\end{equation}

The transactions matrix in Table \ref{table} accounts for all monetary {\em flows} between sectors, expressed in units of currency per time (e.g., USD/year), with a positive entry representing an incoming flow (i.e., a source of funds) and a negative entry representing an outgoing flow (i.e., a use of funds). 
Observe that the ``current'' account for firms registers any monetary payments received or made by the productive sector to all other sectors, including itself, while the ``capital'' account shows how the net investment is funded within the firms' sectors. The behavioural assumptions underlying each of these flows are explained in detail in the next subsection.   

Finally, balance sheets are coupled with transaction flows through the flow of funds matrix in Table \ref{table}, which records the change in stocks resulting from all the monetary transactions occurring at each period of time. 

\begin{table}[htp]
{\small
\begin{center}
\begin{tabular}{|l|c|cc|c|c|}
\hline
 & Households & \multicolumn{2}{|c|}{Firms} & Banks & Row Sum  \\
\hline 
 {\bf Balance Sheets} &  & &  & &  \\
% \hline
Capital stock &  & \multicolumn{2}{|c|}{$pK$} &   & $pK$ \\
%Inventory & & \multicolumn{2}{|c|}{$+cV$} &   & $+cV$ \\
Deposits & $M_h$ & \multicolumn{2}{|c|}{$M_f$} &  $-M$&  0  \\
Loans &  & \multicolumn{2}{|c|}{$-L$} & $L$ &   0 \\
%Bills & $+B$ & & & & $-B$ & 0\\
Equities & ${\mathcal E}_h$  & \multicolumn{2}{|c|}{$-{\mathcal E}_f$}  & $-{\mathcal E}_b$ &   0  \\
\hline
Sum (net worth) & $X_h$ & \multicolumn{2}{|c|}{$X_f=0$}  & $X_b =0 $  & $X$ \\
\hline 
\hline
{\bf Transactions} & &  current & capital &    &\\
% \hline
Consumption  & $-pC_h$ & $pC$ & &    &  0 \\
%Gov Spending & & $+G$ & & & $-G$ & 0 \\
Investment  & & $pI$ & $-pI$ &   & 0  \\
%Change in Inventory & & $+c\dot{V}$ & $-c\dot{V}$ & & 0 \\
%\hline
Accounting memo [GDP] & & [$pY$]  & &  & \\
%\hline
Wages & $W$ & $-W$ & &   & 0 \\
Depreciation & & $-p\bar\delta K$  & $p\bar\delta K$ &  & \\
%Taxes & & $-T$ & & & $+T$ &  0\\
%Subsidies & & $+G_s$ & & & $-G_s$ & 0\\
Interest on deposits  & $r_M M_h$ & $r_M M_f$ & &   $-r_M M$  &   0 \\
Interest on loans  &  & $-r L$ &  & $r L$ &  0 \\
Firms' dividends  & $\Pi_d$ & $-\Pi_d$ & &    & 0 \\
Banks' dividends  & $\Pi_b$ & & &  $-\Pi_b$  & 0 \\
% Profits & $+\Pi_d$ & $-\Pi_d$ &  & &     0\\
\hline
Financial Balances & $S_h$ & $\Pi_r$ & $-pI+p\bar\delta K$ & 0  & 0 \\
%\hline
\hline
\hline
 {\bf Flow of Funds} & &  &    &    &\\
% \hline 
Change in Capital Stock & & \multicolumn{2}{|c|}{$+p\dot K$}& & $+p\dot K$\\
%Change in Inventory & & \multicolumn{2}{|c|}{$+c\dot{V}$}& & $+c\dot{V}$\\
Change in Deposits & $+\dot{M}_h$& \multicolumn{2}{|c|}{$+\dot{M}_f$}&  $-\dot{M}$&   0  \\
Change in Loans & & \multicolumn{2}{|c|}{$-\dot{L}$}   &  $+\dot{ L}$ &   0 \\
% Bills & $-\dot{B}$ & & & & $+\dot{B}$ & 0 \\
%Equities & $+e\dot{E}$ &  \multicolumn{2}{|c|}{$-e\dot{E}$} &    &   0 \\
\hline
Column sum & $S_h$ &  \multicolumn{2}{|c|}{$\Pi_r$}  &   0 & $p\dot K$\\
Change in net worth & $\dot X_h=S_h$ &\multicolumn{2}{|c|}{$\dot X_f=\Pi+\dot{p} K$}   & 0 & $\dot X$\\
\hline
\end{tabular}
\end{center}
\caption{Balance sheet, transactions, and flow of funds for a three-sector economy.}
\label{table}
}
\end{table}%

\subsection{Sector behaviour}

We now describe the behavioural assumptions that govern the evolution of the macroeconomic monetary system.     

% \vspace{0.2in}
% \noindent
% {\bf Production and Financing for Firms}
% \vspace{0.2in}

\subsubsection{Production and Financing for Firms}

Denote real economic output, measured in goods per year, by $Y$ and assume that it is related to the stock of real capital $K$ held by firms, measured in units of goods, through a constant capital-to-output ratio $\bar\nu$ according to 
\begin{equation}
\label{nu}
Y=\frac{K}{\bar\nu} \;,
\end{equation}
so that $\bar\nu$ has units of time. Assumption \eqref{nu} can be relaxed to incorporate a variable rate of capacity utilization as in \cite{GrasselliNguyenHuu2018}, but this generalization is not needed for the purposes of this paper. Capital is then assumed to evolve according to the dynamics
\begin{equation}
\label{eq:capital}
\dot{K} = I - \bar\delta K \;,
\end{equation}
where $I$ denotes gross investment by firms and $\bar\delta$, which has units of inverse of time, is a constant depreciation rate. Firms obtain funds for investment both internally and externally. Following the current account column in Table \ref{table}, we define profits before dividends as\footnote{Observe that we deviate slightly from equation (51) in \cite{GrasselliNguyenHuu2015} by accounting for the cost of depreciation in the definition profits, in line with equation (3) in \cite{BovariGiraudMcIsaac2018}, which we use as our main reference for the empirical estimates of the parameters in the model.} 
\begin{equation}
\label{profit}
\Pi=pY  -W -p\bar\delta K + \bar r_M M_f -r L,
\end{equation}
where $p$ denotes the price level for a basket of goods, measured in units of currency per good, $\bar r_M$ is the instantaneous interest rate received by firms on deposits $M_f$, and $r$ is the instantaneous interest rate paid by firms on loans $L$. 

The deposit rate $\bar r_M$ represents the return on liquid money holdings of households. 
It is the benchmark monetary rate in the real-economy block. 
By contrast, $r$ (and later $\bar r_L$) denotes the baseline funding rate at which firms and financial market participants can borrow in the absence of transaction costs or endogenous risk premia.

Observe that each quantity appearing in \eqref{profit} has units of currency per time (e.g., USD/year), and consequently so does $\Pi$. Accordingly, we define the profit ratio before dividends as the dimensionless quantity 
\begin{equation}
\label{eq:profit share}
\pi=\frac{\Pi}{pY}=1-\omega -\bar\delta\bar\nu + \bar r_M m_f -r \ell ,
\end{equation}
where
\begin{equation}
\label{state_var}
\omega=\frac{W}{pY}, \qquad m_f = \frac{M_f}{pY}, \qquad \ell=\frac{L}{pY}, 
\end{equation}
denote, respectively, the wage share, the deposit ratio for firms, and the loan ratio. Observe that, similar to $\pi$, the wage share $\omega$ is also dimensionless, whereas the ratios $m_f$ and $\ell$ are measured in units of time, since the numerator in each case has units of currency (e.g.,  $M_f$ is measured in USD) whereas the denominator $pY$ has units of currency per time (e.g., USD/year). The profit ratio $\pi$ is used as the main behavioral variable for the firms, as we assume that gross investment is given by 
\begin{equation}
    I=\kappa(\pi)Y
\end{equation}
where $\kappa(\cdot)$ is an increasing function, which for concreteness we assume to have the form 
\begin{equation}
\label{eqn:kappa}
\kappa (\pi) = \min\big(\bar\kappa_{max},\max\left(\bar\kappa_{\min},\bar\kappa_0 + \bar\kappa_1 \pi\right)\big),
\end{equation}
as adopted in \cite{BovariGiraudMcIsaac2018}, where the constants $\bar\kappa_\cdot$ are all dimensionless. 

Next we assume that firms distribute dividends to households according to\footnote{This is also a departure from the baseline model in \cite{GrasselliNguyenHuu2018}, which we adopt in order to be able to use the parameters estimated in \cite{BovariGiraudMcIsaac2018}.}
\begin{equation}
    \Pi_d = \Delta(\pi) pY,
\end{equation}
where, for concreteness, we take
\begin{equation}
\Delta (\pi) = \min\left(\bar\Delta_{max},\max\left(\bar\Delta_{min},\bar\Delta_0 + \bar\Delta_1 \pi\right)\right), \label{eqn:div_function}
\end{equation}
as proposed in \cite{BovariGiraudMcIsaac2018}, where the constants $\bar\Delta_\cdot$ are all dimensionless. 
Accordingly, internal funds for firms are given by the retained profits 
\begin{equation}
    \Pi_r = \Pi - \Pi_d. 
\end{equation}

We further assume that the external funds available to firms in this model consist solely of loans from the banking sector. As can be seen in Table \ref{table}, whenever nominal net investment $p(I-\bar\delta K)$ differs from retained profits $\Pi_r$, firms change their net borrowing, that is,  
\begin{equation}
\label{debt_keen}
\dot{L}-\dot M_f=p(I-\bar\delta K)-\Pi_r.
\end{equation}
The exact partition of the difference between net investment and retained profits among a change in loans $\dot L$ and a change in deposits $\dot M_f$ depends on the liquidity preferences of the firms. We adopt here the following generalization of the formulation in \cite{GrasselliNguyenHuu2015}: 
\begin{align}
\dot M_f &= pY - W - p\bar\delta K + \bar r_M M_f - (1-\bar\zeta) p(I-\bar\delta K) - \bar\kappa_L L -\Pi_d + F \label{deposits} \\
\dot L & = \bar\zeta p(I-\bar\delta K) + r L - \bar\kappa_L L + F \label{loans}
\end{align}
where $\bar\zeta \in [0,1]$ is the proportion of net investment that is financed by new loans, $\bar\kappa_L\in [0,1]$ is the instantaneous rate of loan repayment, and $F$ is a speculative flow. In the absence of $F$, the rationale for this formulation is straightforward: proceeds from sales $pY$ first increase deposits $M_f$, which are then used to pay wages $W$ to workers and to transfer the depreciation amount $p\bar\delta K$ to the capital account. The remainder, with the addition of earned interest $\bar r_M M_f$, is then used to pay for a fraction $(1-\bar\zeta)$ of net investment $p(I-\bar\delta K)$, as well as to repay an amount $\bar\kappa_L L$ of outstanding loans $L$ and dividends to households. Accordingly, the loan amount $L$ increases by the remaining portion of net investment $\bar\zeta p(I-\bar\delta K)$ and the interest charges $r L$, minus the repayment amount $\bar\kappa_L L$. 

The term $F$ in \eqref{deposits}-\eqref{loans} is one of the main modeling innovations in \cite{GrasselliNguyenHuu2015}, which we explore further in this paper. It corresponds to additional borrowing by firms with the purpose of purchasing {\em existing} assets (e.g., real estate, financial instruments, art work) from {\em other firms}. Accordingly, being transactions within a sector, these purchases are not represented in the matrix shown in Table \ref{table}, which only records transactions between sectors (or, in the case of investment, between the current and capital accounts of the firms' sector). Moreover, these debt-financed purchases increase the loans of the buying firm by exactly the same amount as the deposits of the selling firms, as expressed in \eqref{deposits} and \eqref{loans}. In \cite{GrasselliNguyenHuu2015} the dynamics of $F$ follows what had been proposed originally in \cite{Keen1995}, namely with instantaneous increments $\dot F$ given as a varying fraction of nominal output, where here we adopt the simpler formulation  
\begin{equation}
\label{flow_keen}
F=\Psi(g(\pi)+i(\omega))pY,
\end{equation}
% \begin{equation}
% \label{flow_keen}
% \dot F=\Psi(g(\pi)+i(\omega))pY,
% \end{equation}
where $\Psi(\cdot)$ is an increasing function of the observed growth rate $g(\pi)+i(\omega)$ of nominal economic output $pY$, to be defined in \eqref{eq:growth} and \eqref{inflation} below. For concreteness, we take 
\begin{equation}
\label{eq:Psi}
\Psi (x) = \min\left(\bar\Psi_{\max},\max\left(\bar\Psi_{\min},\bar\Psi_0 + \bar\Psi_1 x\right)\right),
\end{equation}
by analogy with the investment and dividend functions in \eqref{eqn:kappa} and \eqref{eqn:div_function}. 
Observe that, unlike in \eqref{eqn:kappa} and \eqref{eqn:div_function} where all constants are dimensionless, 
the constants $\Psi_1$ in \eqref{eq:Psi} have units of time, since the growth rate $g+i$ has units of inverse of time.

% the same formulation as in \cite{GrasselliNguyenHuu2015}, namely
% \begin{equation}
% \label{eq:Psi}
%     \Psi (x) = \bar\psi_0\left(e^{\bar\psi_2 (x-\bar\psi_1)}-1\right)
% \end{equation}

% \vspace{0.2in}
% \noindent
% {\bf The Labour Market and Inflation}
% \vspace{0.2in}

\subsubsection{The Labour Market and Inflation}

We assume that firms follow a minimally rational behaviour by hiring the required amount of labour $E$ at full capacity, that is 
\begin{equation}
    E = \frac{Y}{a}=\frac{K}{a\bar\nu}
\end{equation}
where $a$ represents average labour productivity, measured in units of goods per year per worker and assumed to grow at a constant rate 
\begin{equation}
    \bar\alpha= \frac{\dot{a}}{a}.  \label{eqn:labour}
\end{equation}
Accordingly, we define the endogenously determined employment rate as the dimensionless quantity 
\begin{equation}
\label{eqn:lambda}
     e = \frac{E}{N}
\end{equation}
for a total workforce $N$, which is also assumed to grow at a constant rate   
\begin{equation}
    \bar\beta = \frac{\dot{N}}{N}. \label{eqn:population}
\end{equation}
Notice that both \eqref{eqn:labour} and \eqref{eqn:population} can be relaxed, for example by adopting a logistic growth for population as in \cite{BovariGiraudMcIsaac2018}, but we will not pursue these generalizations here.    

Next we follow \cite{Desai1973} and \cite{Keen2013} and consider a price-wage dynamics of the form
\begin{align}
\label{wage-1}
\frac{\dot{\mathrm{w}}}{\mathrm{w}}&=\Phi( e)+\bar\gamma i(\omega)\;, \\
i(\omega)&:= \frac{\dot p}{p} =-\bar\eta_p\left[1-\bar\xi\frac{\mathrm{w}}{ap}\right] =\bar\eta_p(\bar\xi\omega-1)
\label{inflation}
\end{align}
for constants $\bar\gamma\in[0,1]$, $\bar\eta_p>0$ and $\bar\xi\geq 1$. 
The first equation states that workers bargain for wages based on the current state of the labour market through a Philips curve $\Phi(\cdot)$ while also take into account the observed inflation rate. 
The dimensionless constant $\bar\gamma$ represents the degree of money illusion, 
with $\bar\gamma=1$ corresponding to the case where workers fully incorporate inflation in their bargaining, 
which is equivalent to the wage bargaining in real terms as assumed in \cite{Keen1995}.
The second equation assumes that the long-run equilibrium price is given by a dimensionless markup $\bar\xi$ times unit labor cost $\mathrm{w}/a$, 
whereas observe prices converge to this through a lagged adjustment of exponential form with a relaxation time $1/\bar\eta_p$. For concreteness, we adopt the linear Philips curve specification in \cite{BovariGiraudMcIsaac2018}:
\begin{equation}
\label{eq:philips}
    \Phi (e) = \bar\Phi_0 + \bar\Phi_1 e,
\end{equation}
for constants $\bar\Phi_0$ and $\bar\Phi_1$ with units of inverse of time. 

% \vspace{0.2in}
% \noindent
% {\bf Banks and Households}
% \vspace{0.2in}

\subsubsection{Banks and Households}

We assume that banks satisfy the demands of firms for loans according to \eqref{loans}, while adjusting the effective interest rate on loans $r$ as described in the Section \ref{sec:stochastic}. Furthermore, bank profits of the form
\begin{equation}
\label{eq:bank_profit}
\Pi_b = rL - \bar r_M M
\end{equation}
are fully transferred to households, so that bank equity ${\mathcal E}_b$ remains constant. This assumption can be modified, for example by considering a constant equity ratio instead as in \cite{GrasselliLipton2019a}, but this is not needed for the purposes of this paper.  

Finally, household consumption plays the role of an accommodating variable in the model, in the sense that it must satisfy
\begin{equation}
\label{consumption_keen}
C_h = C = Y-I=(1-\kappa(\pi))Y,
\end{equation} 
since there are no inventories in the model, which means that total output is implicitly assumed to be sold either as investment or consumption. In other words, 
consumption is fully determined by the investment decisions of firms. This 
shortcoming of the model is highlighted in \citet{NguyenHuuPottier2017} and addressed in \cite{GrasselliNguyenHuu2018}, where inventories are included in the model and consumption and investment are 
independently specified. For this paper, however, we adopt the original specification \eqref{consumption_keen}. 

The overall stock--flow consistency of the model can now be verified by checking that the financial balances of all sectors add up to zero as shown in Table \ref{table}, namely 
\begin{align*}
    S_h + \Pi_r - p(I-\bar\delta K)  =& \,\, W + \bar r_M M + \Pi_d + \Pi_b - pC_h + pC + pI\\
    & - W - p\bar\delta K + \bar r_M M_f - rL - \Pi_d -p I + p\bar\delta K  \\
    &+ rL - \bar r_M M - \Pi_b = 0
\end{align*}

\subsection{Dynamical system}

Defining the speculative flow as $f=F/(pY)$ and using the previous definitions \eqref{state_var} and \eqref{eqn:lambda} it then follows that the model corresponds to the four-dimensional system 
\begin{equation}
\label{keen_ponzi}
\hspace{-0.1cm}
\left\{
\begin{array}{ll}
\dot\omega &= \, \omega\left[\Phi( e)-\bar\alpha-(1-\bar\gamma)i(\omega)\right] \\
\dot e &= \, e\left[g(\pi)-\bar\alpha-\bar\beta\right]  \\
\dot m_f &= \, \pi - (1-\bar\zeta)\bar\nu g(\pi) + (r-\bar\kappa_L) \ell -\Delta(\pi)+ f - \left[g(\pi)+i(\omega)\right]m_f\\
\dot \ell &=\,  \bar\zeta(\kappa(\pi)-\bar\delta\bar\nu)+(r-\bar\kappa_L)\ell + f -\left[g(\pi)+i(\omega)\right]\ell \\
% \dot f &= \, \Psi\left(g(\pi)+i(\omega)\right)-\left[g(\pi)+i(\omega)\right] f
\end{array}\right.
\end{equation}
where 
\begin{equation}
\label{eq:growth}
    g(\pi) := \frac{\dot Y}{Y} = \frac{\kappa(\pi)}{\bar\nu}-\bar\delta 
\end{equation}
is the growth rate of real output $Y$ and $f = \Psi(g(\pi)+i(\omega))$ for $\Psi$ defined in \eqref{eq:Psi}, whereas $\pi$, $\kappa(\pi)$, $\Delta(\pi)$, $i(\omega)$ and $\Phi(e)$ are given respectively by \eqref{eq:profit share}, \eqref{eqn:kappa}, \eqref{eqn:div_function}, \eqref{inflation} and \eqref{eq:philips}.
% with $r =r_t$ being a stochastic process specified in equation \eqref{r_stochastic} below. 

\section{The financial market}
\label{sec:stochastic}

There are competing views on asset price determination. 
While standard theory emphasizes fundamentals \citep{fama1970efficient,lucas1978asset}, alternative approaches stress the role of credit conditions and balance-sheet dynamics \citep{gabaix2021search,bouchaud2022inelastic}. 
In line with this perspective, we posit that speculative credit availability influences asset prices.

For example, the deterministic version of the model proposed in \cite{CorsiSornette2014} states that 
\begin{equation*}
    F=\frac{dM_F}{dt} = \bar c_1 S M_F \,, \quad 
\frac{dS}{dt}= \bar c_2 SM_F ,
\end{equation*}
where $M_F$ is the quantity of credit used for purchasing financial assets. This model is itself motivated by the Quantity Theory of Credit proposed in 
\cite{Werner2012}, according to which the standard equation 
\begin{equation}
\label{quantity}
MV = pY
\end{equation}
relating the money supply $M$, the velocity of money $V$, the price level $p$, and real output $Y$ needs to be modified to take financial transactions into account. In other words, equation 
\eqref{quantity} needs to be replaced by the pair of equations  
\begin{align}
\label{quantity_modified}
M_RV_R&=pY \\
M_FV_F &= S Q_F
\end{align}
where the subscript $R$ refers to transactions that are part of GDP, namely transactions on newly produced goods and services, and the subscript $F$ refers to transactions on 
existing financial assets. Assuming that the velocity $V_F$ is approximately constant, it follows that an increase in the amount of money $M_F$ used in financial transactions leads to an 
increase in $SQ_F$, that is, the product of the price level $S$ of existing financial assets and the number of trades $Q_F$ in them.

We draw inspiration from above by linking the speculative credit flow $F$ to asset price dynamics. 
The flow $F$ finances transactions in an existing financial asset whose price is denoted by $S$. 
The variable $S$ serves as an auxiliary financial state variable capturing speculative valuation conditions. 
Although it may be interpreted as a proxy for aggregate market capitalization, it is not explicitly embedded in the stock--flow consistent accounting matrix. 
The real-economy block remains fully stock--flow consistent, while $S$ influences credit conditions through the feedback mechanism described below.

For clarity, we distinguish between the benchmark monetary rate $\bar r_M$, which represents the return on household deposits, and the baseline funding rate $\bar r_L$, which reflects borrowing conditions in financial markets in the absence of endogenous credit spreads.

% The speculative credit flow $F$ finances transactions in an existing financial asset whose price dynamics are governed by equation~\eqref{eq:S}. 
% The price variable $S$ is introduced as an auxiliary financial state variable capturing speculative valuation conditions. 
% While $S$ may be interpreted as a proxy for aggregate market capitalization or equity valuation, it is not explicitly incorporated into the stock--flow consistent accounting matrix. 
% The real-economy block remains fully stock--flow consistent, whereas $S$ operates as an external financial indicator feeding back into credit conditions (see below).
The dynamics of $S$ are given by:
\begin{equation}
\label{eq:S}
dS_t = \bar r_L S_t dt + \bar\sigma S_t dW_t - \bar J^+ S_{t}\left(dN^+_t-\lambda^+ dt\right) + \bar J^- S_{t}\left(dN^-_t-\lambda^- dt\right),
\end{equation}
where $W_t$ is a standard Brownian motion and $N^+_t$ and $N^-_t$ are nonhomogeneous Poisson processes with time-dependent intensities $\lambda^+ = \bar\lambda^+ f^+$ and 
$\lambda^- = \bar\lambda^- f^-$, where $f = F/(pY)$ denotes the normalized speculative flow, and $f^+ = \max(f,0)$, $f^- = \max(-f,0)$ its positive and negative parts.
We assume $0<\bar J^+<1$ and $\bar J^->0$, so that downward jumps correspond to proportional losses  strictly smaller than 100\%, while upward jumps represent proportional gains.
The parameters $\bar\lambda^\pm$ have dimension of inverse time, ensuring that $\lambda^\pm$ are valid jump intensities.
Equation~\eqref{eq:S} is written in compensated form, so that jump terms enter as $(dN_t^\pm - \lambda^\pm dt)$ and the discounted price process is a local martingale.

In the absence of speculative flows ($f=0$), equation~\eqref{eq:S} reduces to a geometric Brownian motion with drift $\bar r_L$ and volatility $\bar\sigma$, so that the discounted price $\widetilde S_t = e^{-\bar r_L t} S_t$ is a martingale.   

When speculative inflows $f^+>0$ are directed toward asset purchases, the model increases the intensity $\lambda^+$ of downward jumps of proportional size $\bar J^+$. 
These jumps capture the buildup of systemic fragility associated with leveraged asset purchases: as credit-financed demand expands, balance-sheet exposure rises and the probability of abrupt corrections increases. 
The compensator term $\bar J^+ \lambda^+ dt$ offsets the expected jump loss, preserving the local martingale property of the discounted price.

Conversely, when speculative outflows $f^->0$ correspond to asset sales aimed at deleveraging, the intensity $\lambda^-$ of upward jumps of proportional size $\bar J^-$ increases. 
These upward jumps represent short-covering or policy-driven rebounds that may occur during deleveraging episodes. 
The compensator $-\bar J^- \lambda^- dt$ offsets the expected upward jump component, again preserving the local martingale property.

% The motivation for this formulation is that, absent any flow of speculative finance, the stock price follows a standard geometric Brownian motion with volatility $\bar\sigma $ and drift $\bar r_L$, so that the discounted stock price 
% \begin{equation}
%     \widetilde S_t = e^{-\bar r_L t}S_t
% \end{equation} 
% is a martingale.

% The presence of a positive speculative flow $f^+$ directed at purchasing assets in the financial market creates an additional risk of downward jumps of proportional size $\bar J^+$ occurring with a time-dependent intensity $\lambda^+$, so that the higher the speculative flow $f^+$ the more frequent the jumps. 
% Investors are compensated for this additional risk by a positive excess drift $\bar J^+\lambda^+ dt$, so that the discounted stock price remains a local martingale. Conversely, a negative speculative flow $-f^-$ represents sales of stocks in order to repay loans, with the corresponding downward drift $-\bar J^-\lambda^- dt$, which is compensated by the presence of occasional upward jumps of proportional size $\bar J^-$ occurring with a time-dependent intensity $\lambda^-$, so that the discounted stock price remains a local martingale as before. 

To transmit financial market conditions to the real economy, we introduce a trend indicator based on the logarithmic return of the asset, following \cite{majewski2020co}.
To do so, consider first the differential equation for the logarithm of the price $S_t$, namely 
\begin{align}
\label{eq:logS}
d\log S_t = &\left(\bar r_L -\frac{1}{2}\bar\sigma^2 +\bar J^+\lambda^+-\bar J^-\lambda^-\right)dt + \bar\sigma dW_t \\
& \qquad + \log\left(1-\bar J^+\right) dN^+_t +\log\left(1+\bar J^-\right) dN^-_t ,
\end{align}
which we can rewrite as 
\begin{align}
\label{eq:logS rewrite}
d\log S_t  = &\left[\bar r_L -\frac{1}{2}\bar\sigma^2 
+\left(\log(1-\bar J^+)+\bar J^+\right)\lambda^++\left(\log(1+\bar J^-)-\bar J^-\right)\lambda^-\right]dt \\ 
& + \bar\sigma dW_t + \log\left(1-\bar J^+\right)\left(dN^+_t-\lambda^+ dt\right)
+ \log\left(1+\bar J^-\right)\left(dN^-_t-\lambda^- dt\right), \nonumber
\end{align}
We define a mean-reverting trend indicator $\mu_t$ that tracks the predictable component of logarithmic returns while preserving the same stochastic shocks:
\begin{align}
\label{eq:trend}
    d\mu_t = & \bar\eta_\mu 
\left[\bar r_L -\frac{1}{2}\bar\sigma^2 
+\left(\log(1-\bar J^+)+\bar J^+\right)\lambda^++\left(\log(1+\bar J^-)-\bar J^-\right)\lambda^-  - \mu_t\right] dt \\
& + \bar\sigma dW_t + \log\left(1-\bar J^+\right)\left(dN^+_t-\lambda^+ dt\right)
+ \log\left(1+\bar J^-\right)\left(dN^-_t-\lambda^- dt\right) \nonumber
\end{align}
By construction, the trend indicator $\mu_t$ inherits the same Brownian and jump shocks as the asset price process. 
This choice reflects the assumption that banks and financial intermediaries observe and react to realized market fluctuations rather than to a smoothed or filtered signal. 
The lending spread therefore responds directly to financial volatility and crash risk.
Consequently, $(S_t,\mu_t)$ defines a two-dimensional Markov jump--diffusion process.
The parameter $\bar\eta_\mu>0$ governs the speed of adjustment toward the conditional drift of log-returns, while the diffusion and jump components are inherited unchanged from the price process.

We assume that the effective lending rate $r_t$ charged by banks is a decreasing function of the trend indicator $\mu_t$:
\begin{equation}
\label{r_stochastic}
r_t = \rho (\mu_t) = \min\left(\bar r_{\max}, \bar r_L + \bar\rho_1 e^{-\bar\rho_2 (\mu_t - \bar r_L)} \right).
\end{equation}
That is to say, banks always charge a premium 
\begin{equation}
\label{premium}
\mathcal{P}_t = \bar\rho_1 e^{-\bar\rho_2 (\mu_t - \bar r_L)}
\end{equation}
above $\bar r_L$ when extending loans to firms. 
In the asset price dynamics, $\bar r_L$ represents the baseline financing rate for leveraged positions and should not be interpreted as a risk-free arbitrage rate.

The premium $\mathcal{P}_t$ is high when the trend indicator is low, reflecting tightening credit conditions during market stress, and declines as $\mu_t$ increases, approaching zero in periods of sustained asset price growth. 
This specification captures procyclical lending behavior. 
The upper bound $\bar r_{\max}$ ensures boundedness of the rate and may also be interpreted as reflecting policy intervention when private credit markets become severely impaired.

% The premium $X$ is close to the parameter $\bar\rho_1$ when the trend tracker is close to the lower bound rate $\bar r_L$.
% In periods when the trend tracker is high, the premium decreases, approaching zero asymptotically in periods of market exuberance. 
% Conversely, when the trend tracker is low, the premium charged on loans increases rapidly, making it much more difficult for firms to borrow in periods of market turbulence. We assume a maximum interest rate $\bar r_{\max}$ primarily for mathematical convenience, but also reflecting policy interventions when the effective lending rate is so high that private credit market is deemed to be frozen.  

Observe that 
\begin{equation}
\label{logX}
    \log \mathcal{P}_t = -\bar\rho_2 \mu_t + \bar\rho_2\bar r_L+ \log \bar\rho_1,
\end{equation}
so we can readily derive properties of the premium process $X_t$ from those obtained for the trend indicator $\mu_t$.

Notice that the absolute level of $S_t$ is economically irrelevant and can be re-normalized without affecting the macro-financial dynamics, as only its relative movements influence credit conditions.

The full macro-financial system is therefore described by the coupled stochastic dynamics of 
$(\omega_t, e_t, b_t, f_t, S_t, \mu_t)$, 
where the real-economy variables satisfy deterministic differential equations driven by the stochastic lending rate $r_t$, while $(S_t,\mu_t)$ evolve as a compensated jump--diffusion process. 

Under the stated parameter restrictions, the coupled process defines a time-inhomogeneous Markov process on the admissible state space
\[
\mathcal{D} = \{(\omega,e,b,f,S,\mu) : \omega \ge 0,\ e \ge 0,\ S>0,\ (b,f,\mu)\in\mathbb{R}\}.
\]
The restriction $0<\bar J^+<1$ ensures strict positivity of the asset price.

Existence and uniqueness of strong solutions follow from standard arguments for stochastic differential equations with locally Lipschitz coefficients and state-dependent jump intensities, and are established in Appendix~\ref{sec:existence}. 
Under mild growth conditions on the macroeconomic feedback functions, the system does not exhibit finite-time explosion. 

Consequently, the model defines a well-posed stochastic dynamical system suitable for analytical and numerical investigation.

\subsection{Special cases and fundamental properties}

To recap, we connect the real economy model \eqref{keen_ponzi} to the stock market model \eqref{eq:S} through two channels: (1) the intensities of jumps $\lambda^\pm = \bar\lambda^\pm f^\pm$, which add boom-and-bust dynamics to the stock price in proportion to the speculative flow $f$ and (2) the effective interest rate on loans in \eqref{r_stochastic}, which connects the risk premium on loans provided to firms with the trend indicator for the log returns of stock prices. 

Accordingly, the real economy model becomes independent from the stock market when $\bar\rho_1 = 0$, in which case  $r = \bar r_L$ and the stock price has no influence in \eqref{keen_ponzi}. Conversely, the stock market becomes independent from the real economy when $\bar\lambda^\pm = 0$, in which case the speculative flow $f$ has no influence in \eqref{eq:S}.

\subsubsection{The deterministic case}
Before introducing stochasticity, one can already analyze the determinisic extended model
with $\bar\sigma=\bar\lambda^\pm=0$.
In this case, the stock price process becomes the growing exponential $S_t = S_0 e^{\bar r_L t}$ and the trend indicator satisfies 
\[
\lim_{t\to +\infty} \mu_t = \bar r_L
\]
and, as a consequence, 
$
\lim_{t\to +\infty} r_t = \bar r_L + \bar\rho_1
$.

% This shows some internal consistency of the stock price dynamics at the macroeconomic scale:
% as the financial market shows no sign of risk, absence of arbitrage must
% align the lending rate $r_t$ with the return rate of the asset.
% Asymptotically, it becomes $\bar r_M$ which writes as a risk-free rate.
% The above remark is of course crucially relying on $\bar r_M\ge 0$, which appears as a natural assumption.

\subsubsection{The case $\bar\lambda^\pm =0$}
\label{sec:brownian case}

In this case, the stochastic differential equation \eqref{eq:S} has a solution given by %\footnote{see Appendix \ref{sec:existence} for construction}
\begin{equation}
\label{eq:S_solved}
S_t = S_0 e^{ \left(\bar r_L - \frac{1}{2}\bar\sigma^2\right)t + \bar\sigma W_t } \;, \quad t\ge 0 \; ,
\end{equation}
whereas the trend indicator satisfies 
\begin{equation}
\label{eq:trend_nojump}
    d\mu_t =  \bar\eta_\mu 
\left(\bar r_L - \frac{1}{2}\bar\sigma^2 -\mu_t\right) dt + \bar\sigma dW_t,
\end{equation}
which defines an Ornstein–Uhlenbeck (OU) process with long-term mean 
\begin{equation}
\label{mu0}
    \bar\mu_0 = \bar r_L - \frac{1}{2}\bar\sigma^2,
\end{equation}
mean-reversion speed $\bar\eta_\mu$ and volatility $\bar\sigma$. Consequently, in this case, the trend indicator is asymptotically stationary with a Gaussian distribution with mean 
$\bar\mu_0$ and variance $\frac{\bar\sigma^2}{2\bar\eta_\mu}$. 
% where $W$ is a Brownian motion for a probability $\P^W$ on a properly given space, see Appendix \ref{sec:existence}.
% For $\bar r_M\ge 0$, it is a continuous positive submartingale, and the discounted
% process $\tilde{S} = exp(-\bar r_M t)S_t$ is a $\P^W$-martingale.
Accordingly, the premium process $X_t$ defined in 
\eqref{premium} is asymptotically log-normal with parameters $\left(\frac{\bar\rho_2\bar\sigma^2}{2},\frac{\bar\rho_2^2\bar\sigma^2}{2\bar\eta_\mu}\right)$, so that its asymptotic mean is equal to 
\begin{equation}
    \bar X_{0} = \bar\rho_1 e^{\frac{\bar\rho_2\bar\sigma^2}{2}+\frac{\bar\rho_2^2\bar\sigma^2}{4\bar\eta_\mu}} > \bar\rho_1,
\end{equation}
meaning that, on average, the asymptotic effective rate $r_t$ remains above the base level $\bar r_L + \bar\rho_1$ obtained when $\bar\sigma = 0$, and this difference is increasing in $\bar\rho_2$ and $\bar\sigma$, but decreasing in $\bar\eta_\mu$.

\subsubsection{The case $f^+ = \bar f >0$}
\label{f_positive}

In this case $f^- = 0$ and the trend indicator $\mu_t$ defined in \eqref{eq:trend} becomes an extended Ornstein–Uhlenbeck (OU) proces with Poisson jumps of constant size $\log(1-\bar J^+)$ and constant intensity $\bar\lambda^+\bar f>0$, which has an asymptotic stationary distribution (see, for example, \cite{Applebaum04}) consisting of the convolution of a Gaussian and compound Poisson distributions whose characteristic function $\phi (u)$ is given by 
\begin{equation}
    \log \phi (u) = i u\bar\mu^+ -\frac{\bar\sigma^2 u^2}{4\bar\eta_mu}+\bar\lambda^+\bar f \int_0^\infty \left( e^{i u \tilde J^+e^{-\bar\eta_u s}}-1-i u \tilde J^+ e^{-\bar\eta_\mu s}\right)ds,
\end{equation}
where 
\begin{equation}
\label{mu_plus}
    \bar\mu^+ = \bar r_L -\frac{1}{2}\bar\sigma^2 
+\left(\tilde J^++\bar J^+\right)\bar\lambda^+\bar f,
\end{equation}
and 
\begin{equation}
    \tilde J^+ = \log (1-\bar J^+).
\end{equation}
From it, one can deduce that the mean of the stationary distribution is given by $\bar \mu^+$, whereas its variance is given by 
\begin{equation}
\label{var_infty}
\text{Var}(\mu)_\infty = \frac{\bar\sigma^2}{2\bar\eta_\mu} + \frac{\bar\lambda^+\bar f (\tilde J^+)^2}{2\bar\eta_\mu}.
\end{equation}
That is to say, compared to the previous case with $\bar\lambda^\pm = 0$, the case with constant intensity for downward jumps reduces the mean of the stationary distribution of the trend indicator to $\bar\mu^+ < \bar\mu_0$, since $\left(\log(1-\bar J^+)+\bar J\right)\bar\lambda^+\bar f <0$, and increases its variance. Moreover, the skewness of the stationary distribution of the trend indicator is 
\begin{equation}
\label{skew}
    \text{Skewness}(\mu)_\infty = -\frac{2^{3/2}\bar\lambda^+\bar f (\tilde J^+)^3}{3\left(\bar\sigma^2+\bar\lambda^+\bar f (\tilde J^+)^2\right)^{3/2}} > 0,
\end{equation}
as the jumps of size $\log(1-\bar J^+)<0$ make it spend more time below its mean. Finally, its kurtosis is given by 
\begin{equation}
\label{kurtosis}
    \text{Kurtosis}(\mu)_\infty = 3+ \frac{4\bar\eta_\mu\bar\lambda^+\bar f(\tilde J^+)^4}{\left(\bar\sigma^2+\bar\lambda^+\bar f (\tilde J^+)^2\right)^2} > 3,
\end{equation}
meaning that the jumps cause it to be fat-tailed. All of these properties translate easily to the asymptotic stationary distribution of the process $\log X_t$ in \eqref{logX} and, with the necessary adjustment, to the corresponding lending rate in \eqref{r_stochastic}. 

\subsubsection{The case $f^- = \bar f >0$}
\label{f_negative}

In this case $f^+ = 0$ and the trend indicator $\mu_t$ defined in \eqref{eq:trend} becomes an extended Ornstein–Uhlenbeck (OU) process with Poisson jumps of constant size $\log(1+\bar J^-)$ and constant intensity $\bar\lambda^-\bar f>0$. The same calculations as before lead to an asymptotic stationary distribution with mean 
\begin{equation}
    \bar\mu^- = \bar r_L -\frac{1}{2}\bar\sigma^2 
+\left(\tilde J^--\bar J^-\right)\bar\lambda^-\bar f,
\end{equation}
where
\begin{equation}
    \tilde J^- = \log (1+\bar J^-).
\end{equation}
The variance, skewness and kurtosis of the asymptotic stationary distribution are given by the same expressions as in \eqref{var_infty}, \eqref{skew} and \eqref{kurtosis}, but with $\bar\lambda^+$ replaced with $\bar\lambda^-$ and $\tilde J^+$ replaced with $\tilde J^-$. 

Compared to the case with $\bar\lambda^\pm = 0$, we see that the case with constant intensity for upward jumps also reduces the mean of the stationary distribution of the trend indicator to $\bar\mu^- < \bar\mu_0$, since $\left(\log(1+\bar J^-)-\bar J^-\right)\bar\lambda^-\bar f <0$, and increases its variance. Moreover, the skewness is now negative, as the jumps of size $\log (1+\bar J^+)>0$ make the trend indicator spend more time above its mean, whereas the distribution is also fat-tailed, as the kurtosis is greater than 3.

\section{Numerical Simulations}
\label{sec:numerics}

The exogenously fixed parameters in the model are given in Table \ref{parameters}, where we specify the theoretical range of each parameter, as well as the baseline value used in the simulations, listed in the order in which they appear in the text. The base values for the parameters $\bar\nu$, $\bar\delta$, $\bar\kappa_{\cdot}$, $\bar\alpha$ and $\bar\Phi_{\cdot}$ are taken from \cite{BovariGiraudMcIsaac2018}, whereas those for $\bar\Delta_{\cdot}$, $\bar\eta_p$ and $\bar\xi$ are taken from the update provided in \cite{bovari2018debt} (which otherwise uses the same values as \cite{BovariGiraudMcIsaac2018}). All other base values were chosen by us. 

Unless otherwise specified, parameter values are chosen so as to satisfy the theoretical restrictions ensuring global existence and non-explosion established in Section~\ref{sec:existence}. 
The numerical experiments therefore explore the dynamic implications of the model within the mathematically admissible region of the parameter space.

\begin{table}[htbp]
\renewcommand*{\arraystretch}{1.2} %% a little more vertical space
\tiny
\caption{Model parameters}
\centering
\begin{tabular}{lllll} 
\hline
Symbol & Range & Base Value &  Description & Equation \\
\hline
$\bar\nu$ & $(0,\infty)$ & 2.7   &  Capital to output ratio & \eqref{nu} \\
$\bar\delta$ & $[0,\infty)$ & 0.04 & Depreciation rate of capital & \eqref{eq:capital} \\
$\bar r_M $ & $[0,\infty)$ & 0.01 & Interest on deposits & \eqref{profit} \\
$\bar\kappa_{min}$ & $[-1,0]$ & 0 & Investment  function minimum   & \eqref{eqn:kappa}\\
$\bar\kappa_{max}$ & $[0,1]$ & 0.3 & Investment  function maximum   & \eqref{eqn:kappa}\\
$\bar\kappa_0$ & $\mathbb{R}$ & 0.0318 & Investment function intercept & \eqref{eqn:kappa}\\ 
$\bar\kappa_1$ & $\mathbb{R}$ & 0.575 & Investment function slope & \eqref{eqn:kappa}\\
$\bar\Delta_{min}$ & $[-1,0]$ & 0 & Dividend  function minimum & \eqref{eqn:div_function} \\
$\bar\Delta_{max}$ & $[0,1]$ & 0.3 & Dividend  function maximum & \eqref{eqn:div_function} \\
$\bar\Delta_0$ & $\mathbb{R}$ & -0.078 & Dividend function intercept & \eqref{eqn:div_function} \\
$\bar\Delta_1$ & $[0,\infty)$ & 0.553 & Dividend  function slope 
& \eqref{eqn:div_function} \\
$\bar\zeta$ & $[0,1]$ & 0.8 & Proportion of net investment financed by new loans 
& \eqref{deposits} \\
$\bar\kappa_L$ & $[0,1]$ & 0.02 & Instantaneous rate of loan repayment  
& \eqref{deposits} \\
$\bar\Psi_{\min}$ & $[-1,0]$ & -0.15 & Speculative flow minimum   
& \eqref{eq:Psi} \\
$\bar\Psi_{\max}$ & $[0,1]$ & 0.3 & Speculative flow maximum   
& \eqref{eq:Psi} \\
$\bar\Psi_0$ & $\mathbb{R}$ & -0.075 & Speculative flow intercept   
& \eqref{eq:Psi} \\
$\bar\Psi_1$ & $[0,\infty)$ & 3.75 & Speculative flow slope    
& \eqref{eq:Psi} \\
$\bar\alpha$ & $\mathbb{R}$ & 0.02 & Growth rate of productivity & \eqref{eqn:labour} \\ 
$\bar\beta$ & $\mathbb{R}$ & 0.02 & Growth rate of workforce & \eqref{eqn:population} \\ 
$\bar\gamma$ & $\mathbb{R}$ & 0.9 & Effect of inflation on wages & \eqref{wage-1} \\
$\bar\eta_p$ & $(0,\infty)$ & 0.192 & Speed of adjustment of price index & \eqref{inflation} \\
$\bar\xi$ & $[1,\infty)$& 1.875 & Price markup  & \eqref{inflation} \\
$\bar\Phi_0$ & $\mathbb{R}$ & -0.292 & Phillips curve intercept & \eqref{eq:philips}  \\
$\bar\Phi_1$ & $\mathbb{R}$ & 0.469 & Phillips curve slope  & \eqref{eq:philips} \\
$\bar r_L$ & $[0,\bar r_{\max}]$ &0.02 & Long term interest rate & \eqref{eq:S} \\
$\bar \sigma$ & $[0,\infty)$ & 0.1 & Volatility of stock price & \eqref{eq:S} \\
$\bar J^+$ & $[0,1]$ & 0.1 & Proportional downward jump in stock price & \eqref{eq:S} \\
$\bar J^-$ & $[0,1]$ & 0.1 & Proportional upward jump in stock price & \eqref{eq:S} \\
$\bar \lambda^+$ & $[0,\infty)$ & 1 & Intensity parameter for downward jump in stock price & \eqref{eq:S} \\
$\bar \lambda^-$ & $[0,\infty)$ & 1 & Intensity parameter for upward jump stock price jump & \eqref{eq:S} \\
$\bar \eta_\mu$ & $[0,\infty)$ & 0.5 & Speed of adjustment of market trend & \eqref{r_stochastic} \\
$\bar r_{\max}$ & $[0,\infty)$ &0.2 & Maximum interest rate & \eqref{r_stochastic} \\
$\bar\rho_1$ & $[0,\infty)$ &0.01 & Parameter in interest rate function  & \eqref{r_stochastic} \\
$\bar\rho_2$ & $[0,\infty)$ &5 & Parameter in interest rate function  & \eqref{r_stochastic} \\
\hline
\end{tabular}
\label{parameters}
\end{table}

All simulations are implemented in Julia (v1.10+) \citep{bezanson2017julia}, using the 
package
\texttt{DifferentialEquations.jl} \citep{rackauckas2017differentialequations}. 
The model is formulated as a nonlinear system of stochastic differential equations with state-dependent jumps (Jump-SDE). Continuous dynamics are discretized using strong order-1 stochastic Runge–Kutta scheme (\texttt{SRIW1}), with automatic fallback to 
Stability-optimized adaptive strong order scheme (\texttt{SOSRI}) in rare stiff regimes. 
Jump components are handled via \texttt{JumpProcesses.jl}, using either constant or variable intensity Poisson processes, depending on the specification \citep{rackauckas2021sde}. 
Jump intensities are state-dependent and evaluated endogenously at each step, while jump impacts are implemented as multiplicative shocks to asset prices and additive shocks to auxiliary state variables. 
Numerical stability is enforced through domain checks (positivity, finiteness)\footnote{Code available on Github:
\url{https://github.com/AdrienNguyenHuu/research-projects/tree/main/SFC-stochastic-bubble}.}.

\subsection{Representative trajectories}
\label{sec:equilibrium}

\subsubsection{Decoupled cases}

We start with the fully decoupled case when $\bar\rho_1 = \bar\lambda^\pm = 0$. In this case, the real economy system \eqref{keen_ponzi} exhibits behaviour similar to that in \cite{GrasselliNguyenHuu2015}, with convergence to an interior equilibrium for the baseline parameters and divergence to an explosive regime when the cost of borrowing, here represented by the difference $(\bar r_L-\bar r_M)$, increases, as shown in Figure \ref{fig:decoupled_a_b}. As expected, in both cases the stock price behaves as a geometric Brownian motion, regardless of the state of the economy, with the trend indicator mean-reverting to the long-term value $\bar\mu_0$ defined in \eqref{mu0}

\begin{figure}[ht!]
\includegraphics[width=.49\linewidth]{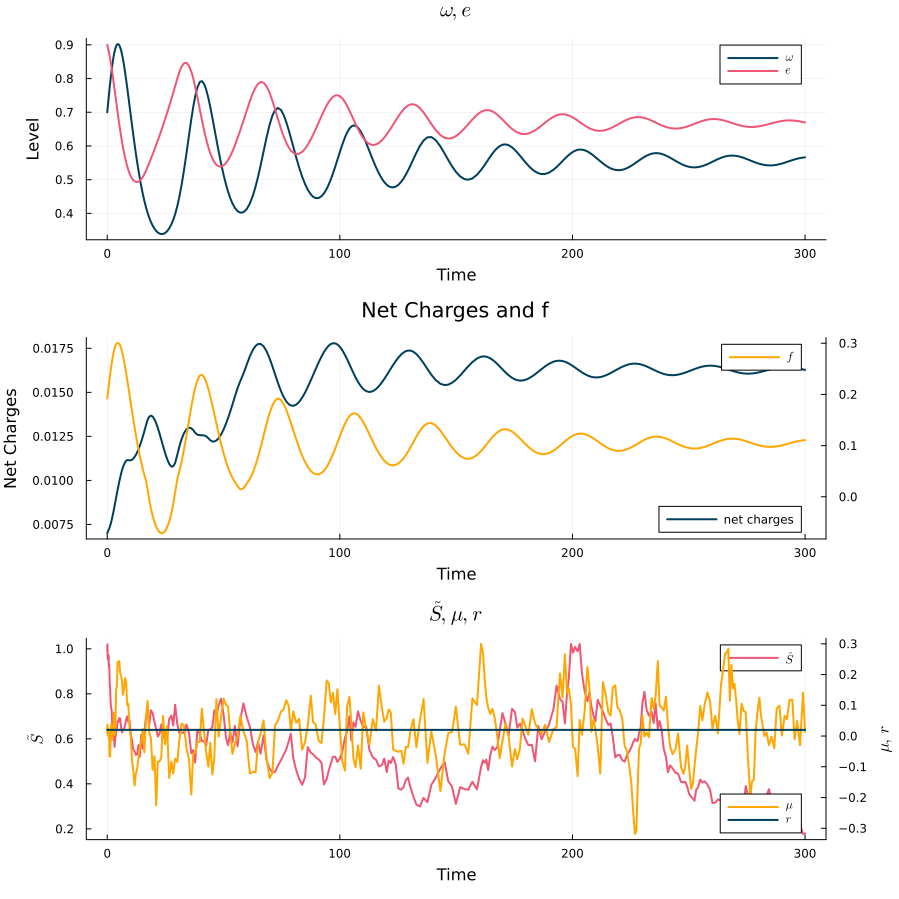}\hfill
\includegraphics[width=.49\linewidth]{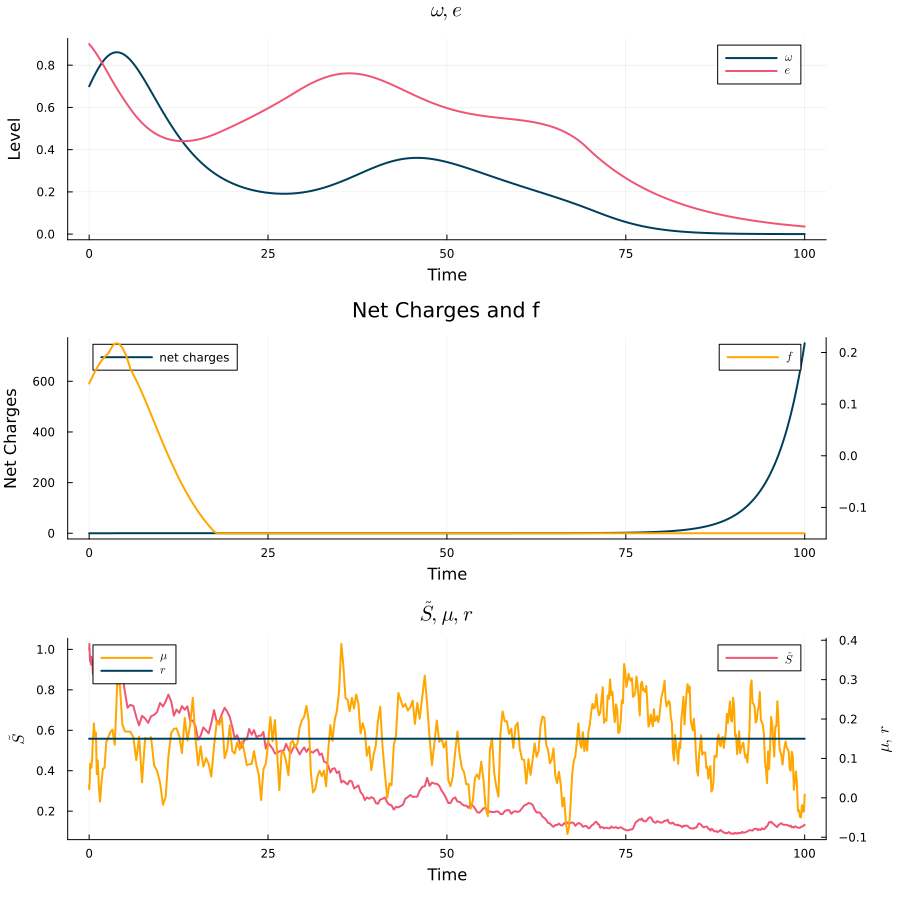}
\caption{\small Wage share $\omega$ and employment rate $e$ (top row), net financial charges $\bar r_L \ell -\bar r_M m$ and speculative flow $f$ (middle row), and discounted stock price $\tilde S$, trend indicator $\mu$ and effective rate $r$ (bottom row) in the deterministic case with $\bar\rho_1 = \bar\lambda^{\pm}=0$. A low interest rate $\bar r_L=0.02$ (left) leads to the interior equilibrium for economic variables and $\bar\mu_0 =  0.025$ for the long-term mean of the trend indicator. A higher interest rate $\bar r_L=0.15$ (right) leads to a collapse of the real economy and a higher value $\bar\mu_0 = 0.145$ for the long-term mean of the trend indicator. In both cases the stock price follows a geometric Browning motion.}\label{fig:decoupled_a_b}
\end{figure}
 
Next we consider the same two regimes, namely low and high cost of borrowing, but with nonzero jump intensities $\bar\lambda^+$ and $\bar\lambda^-$, while keeping $\rho_1 = 0$, so that the real economy affects the stock market but not the other way around. As expected, the real economy behaves in exactly the same way as in the previous case, but now has an effect on the stock price, as shown in Figure \ref{fig:decoupled_c_d}. In the convergent case, we see that the speculative flow is asymptotically constant at a positive value 
$f^+ = \bar f \approx 0.0566$, so that we are in the situation described in Section \ref{f_positive}, with downward jumps in the stock price and a positive drift, leading to the long-term average $\bar\mu^+<\bar\mu_0$ defined in \eqref{mu_plus} for the trend indicator. Conversely, in the divergent case we see that the speculative flow is asymptotically constant at a negative value with $f^- = \bar f \approx 0.15$, so that we are in the situation described in Section \ref{f_negative}, with upwardd jumps in the stock price and a negative drift, leading to the long-term average 
$\bar\mu^- <\bar\mu_0$ for the trend indicator. 

\begin{figure}[ht!]
\includegraphics[width=.49\linewidth]{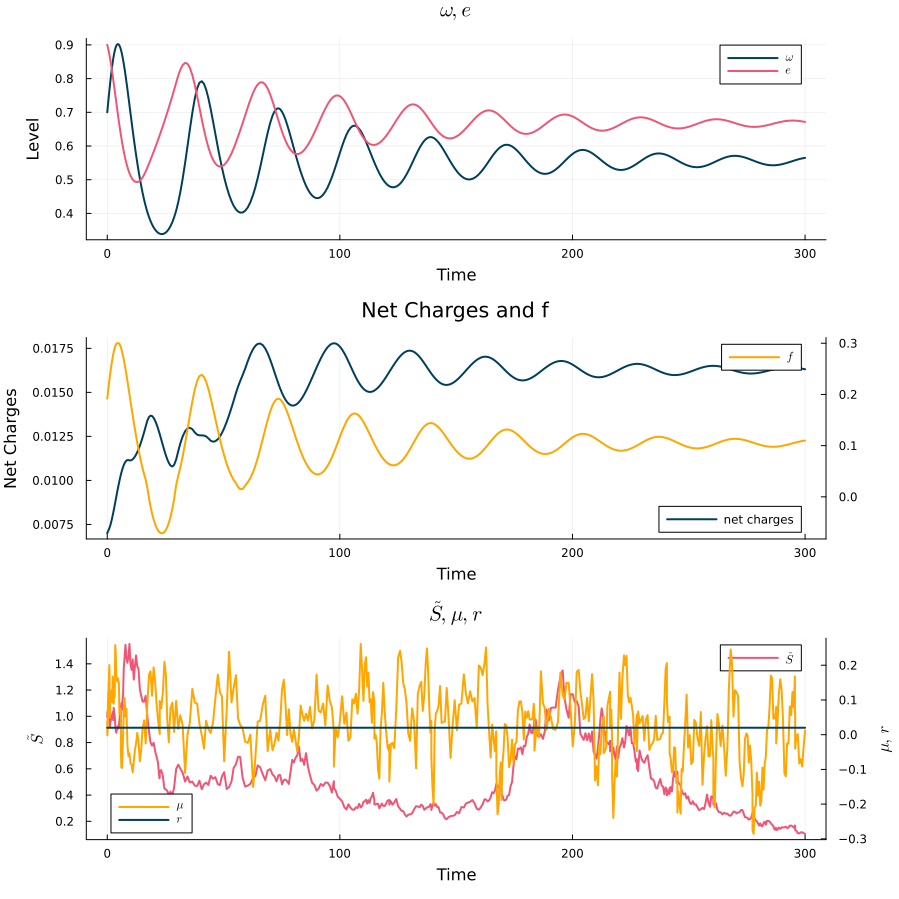}\hfill
\includegraphics[width=.49\linewidth]{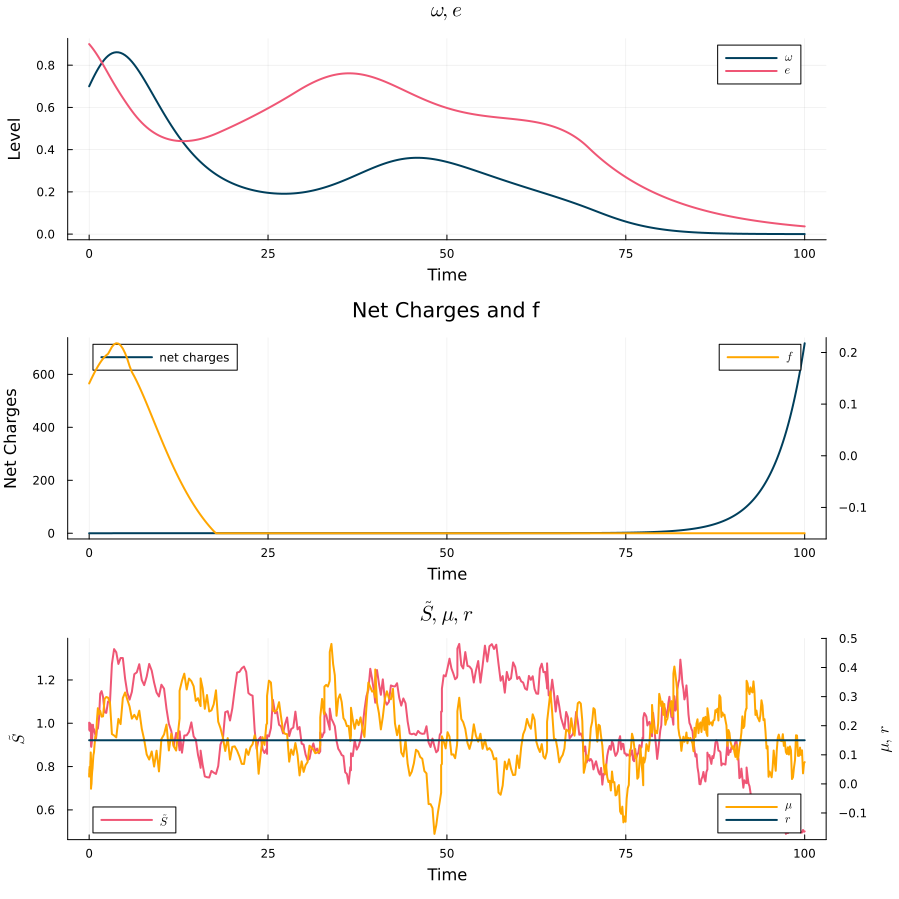}
\caption{\small Wage share $\omega$ and employment rate $e$ (top row), net financial charges $\bar r_L \ell -\bar r_M m$ and speculative flow $f$ (middle row), and discounted stock price $\tilde S$, trend indicator $\mu$ and effective rate $r$ (bottom row) in the case with nonzero jump intensities $\bar\lambda^\pm = 1$ but $\rho_1 = 0$. A low interest rate $r_L=0.02$ (left) still leads to the interior equilibrium for economics variables, but stock price now has downward jumps of relative size $\bar J^+ = 0.1$ occurring on average once every $(\bar\lambda^+\bar f)^{-1}\approx 17.67$ years, and $\bar\mu^+ \approx 0.02469 < \bar\mu_0$. A high interest rate $r_L=0.15$ still leads to a collapse of the economy, with the stock price exhibiting upward jumps of relative size $\bar J^- = 0.1$ occurring on average  once every $(\bar\lambda^-\bar f)^{-1}\approx 6.67$ years, and $\bar\mu^- \approx  0.1435 < \bar\mu_0 $. (right)}\label{fig:decoupled_c_d}
\end{figure}

As a final experiment in this section we revert to the case of jump intensities $\bar\lambda^\pm = 0$ but allow the stock market to influence the economy through the interest rate by setting $\bar\rho_1 = 0.01$ in \eqref{r_stochastic}. In this case, even in the absence of jumps, we see that the stock market influence through the effective interest rate $r_t$ can have the effect of causing the real economy to collapse even in the case of a low baseline interest rate $\bar r_L = 0.02$, essentially by raising the effective lending rate during periods of market downturn.

\begin{figure}[ht!]
\includegraphics[width=.49\linewidth]{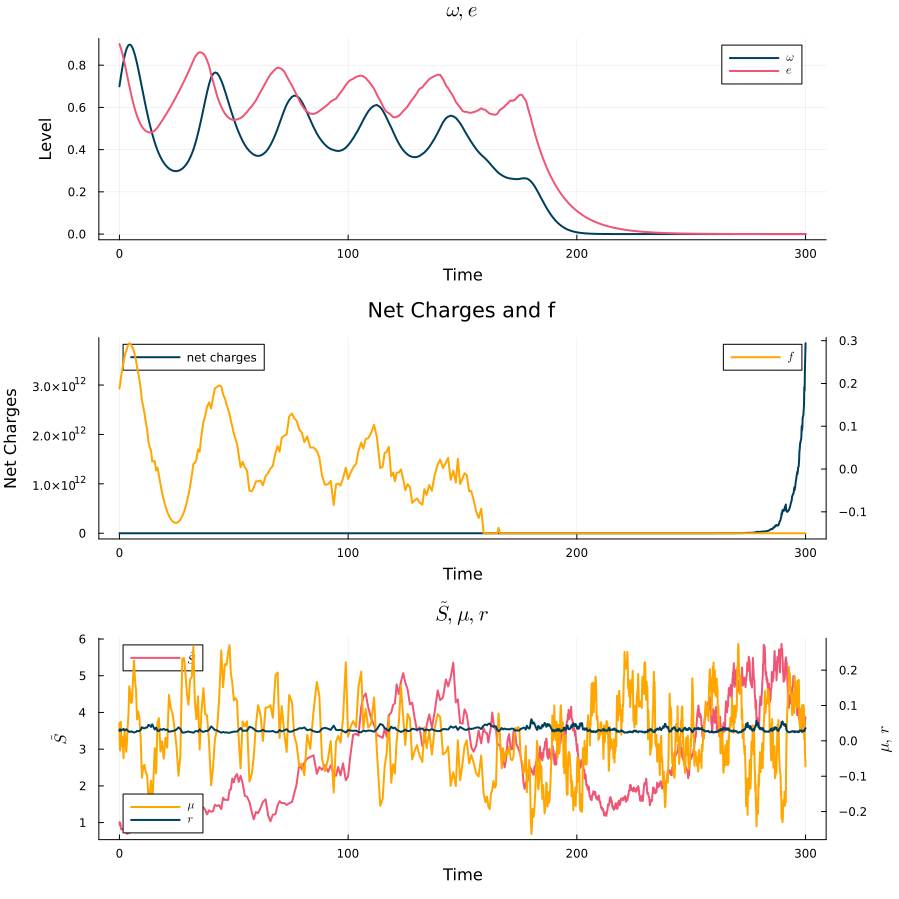}\hfill
\includegraphics[width=.49\linewidth]{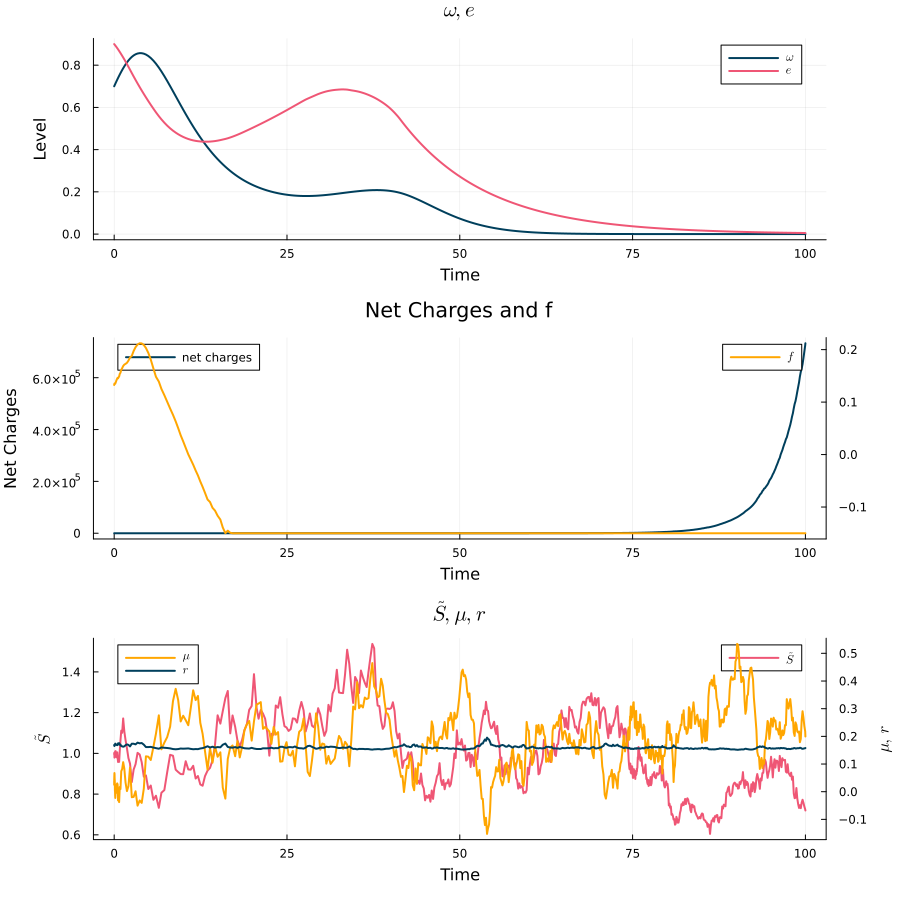}
\caption{\small Wage share $\omega$ and employment rate $e$ (top row), net financial charges $\bar r_L \ell -\bar r_M m$ and speculative flow $f$ (middle row), and discounted stock price $\tilde S$, trend indicator $\mu$ and effective rate $r$ (bottom row) in the case with $\bar\lambda^\pm = 0$ and $\rho_1 = 0.01$. A low baseline interest rate $r_L=0.02$ (left) can now be associate with a collapse in the economy because of the influence of the stock price on the effective rate $r$, even in the absence of jumps. The higher baseline interest rate $r_L=0.15$ (right) simply makes this collapse occur sooner.}\label{fig:decoupled_e_f}
\end{figure}

\subsubsection{Fully coupled cases}

We now consider the full model, where the economy affects the stock market through nonzero jump intensities that depend on the speculative flow $f$ and the stock market affect the economy through the stochastic interest $r$. We consider first the baseline parameters in Table 1, which corresponds to adding a stochastic interest rate to the case previously considered in the left panel of Figure \ref{fig:decoupled_c_d} or, alternatively, adding stock price jumps to the case previously considered in the left panel of Figure \ref{fig:decoupled_e_f}. What we observe is that the economy fluctuates between periods of high and low speculative flows and eventually collapses when interest rates remain high for extended periods, as shown in the left panel of Figure \ref{fig:coupled_a_b}. The right panel shows similar behaviour, only exacerbated by a higher stock price volatility $\bar\sigma = 0.25$, instead of the base value $\bar\sigma = 0.1$ from Table \ref{parameters}. 

\begin{figure}[ht!]
\includegraphics[width=.49\linewidth]{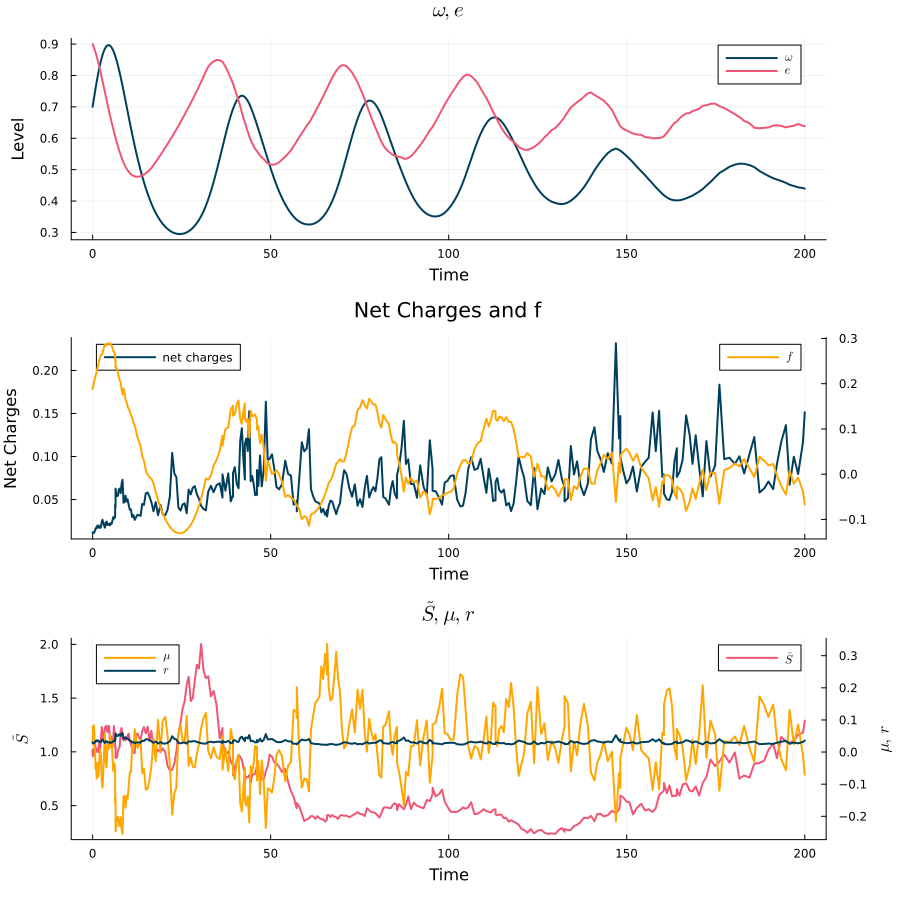}\hfill
\includegraphics[width=.49\linewidth]{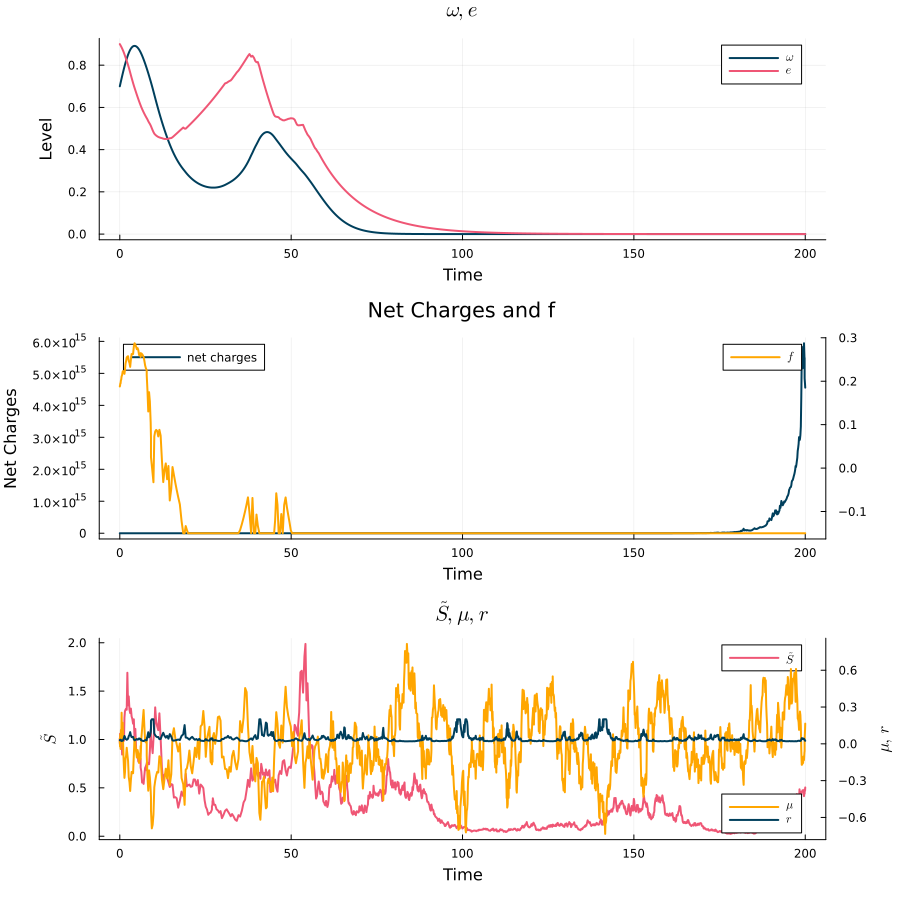}
\caption{\small Wage share $\omega$ and employment rate $e$ (top row), net financial charges $\bar r_L \ell -\bar r_M m$ and speculative flow $f$ (middle row), and discounted stock price $\tilde S$, trend indicator $\mu$ and effective rate $r$ (bottom row) for baseline case with all parameters as in Table \ref{table} (left) and case with higher stock price volatility $\bar\sigma = 0.25$ (right).} \label{fig:coupled_a_b}
\end{figure}

Next in Figure \ref{fig:coupled_c_d}, we explore the effect of the speed of mean reversion of the trend indicator $\mu$. Instead of the base value $\bar\eta_\mu = 0.5$ from Table \ref{parameters}, the left panel shows the results for a faster mean reversion $\bar\eta_\mu = 5$, where we see that the economy experiences prolonged periods without a major crash, whereas the right panel shows that a slower mean reversion $\bar\eta_\mu = 0.2$ leads to a crash sooner.    

\begin{figure}[ht!]
\includegraphics[width=.49\linewidth]{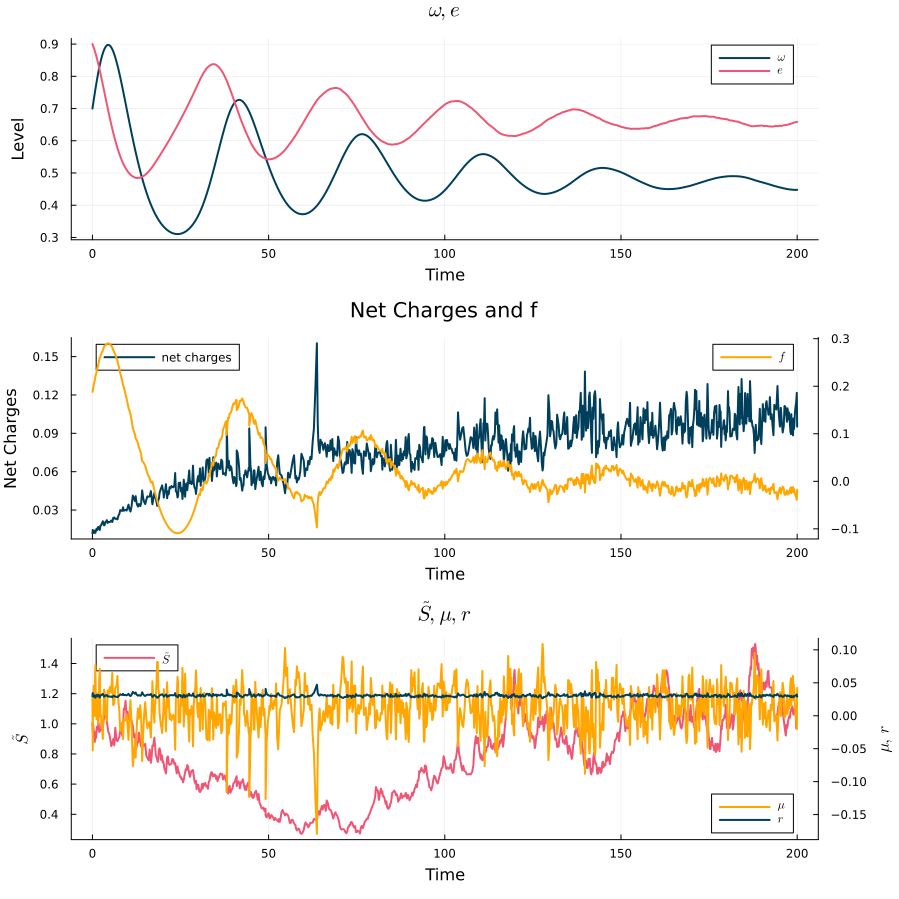}\hfill
\includegraphics[width=.49\linewidth]{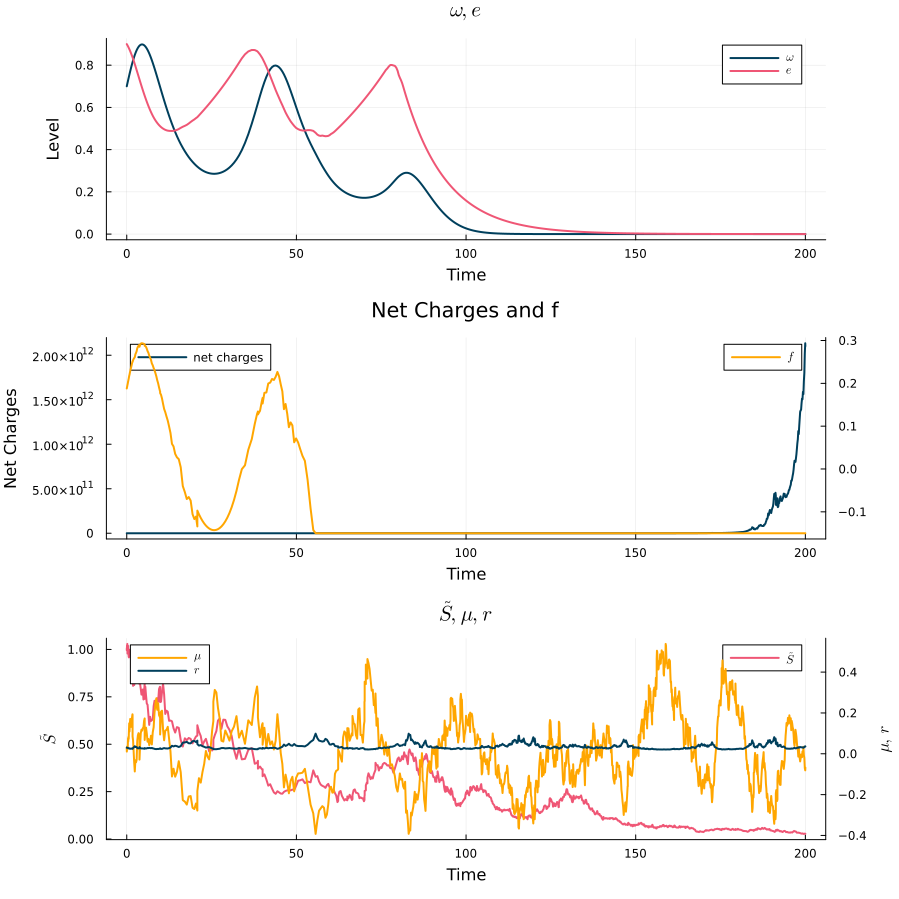}
\caption{\small Wage share $\omega$ and employment rate $e$ (top row), net financial charges $\bar r_L \ell -\bar r_M m$ and speculative flow $f$ (middle row), and discounted stock price $\tilde S$, trend indicator $\mu$ and effective rate $r$ (bottom row) for the cases with $\bar\eta_\mu = 5$ (left) and $\bar\eta_\mu = 0.2$.} \label{fig:coupled_c_d}
\end{figure}

As a final experiment in this section, in Figure \ref{fig:coupled_e_f} we repeat the conditions of faster and slower mean reversion of the trend indicator, but with $\bar\rho_2 = 3$ instead of the base value $\bar\rho_2 = 5$ from Table \ref{parameters}, meaning that the response of the interest rate to a drop in the trend indicator is less pronounced, with corresponding improvements in the economic outcomes. 

\begin{figure}[ht!]
\includegraphics[width=.49\linewidth]{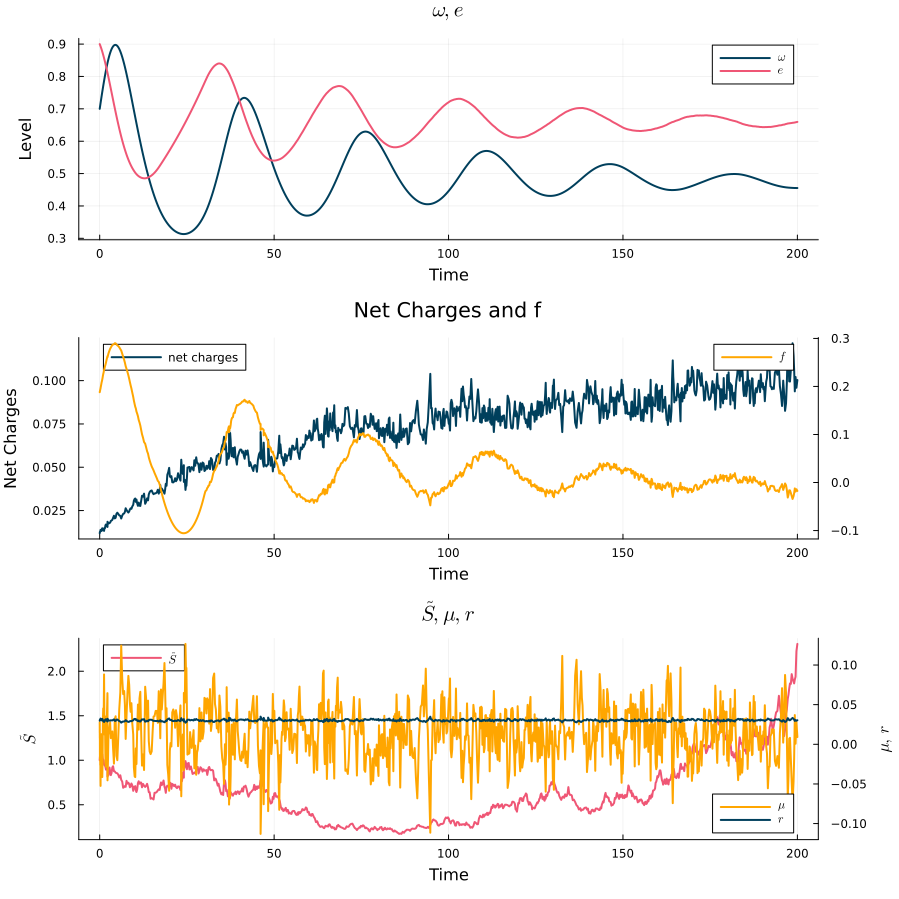}\hfill
\includegraphics[width=.49\linewidth]{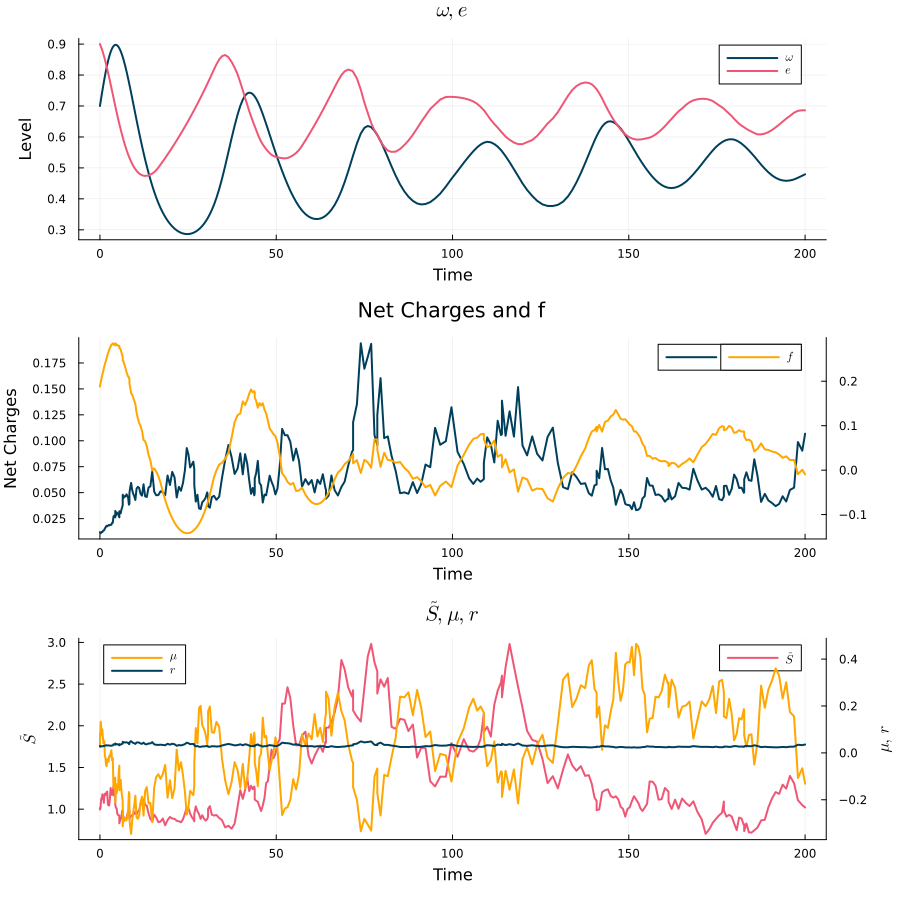}
\caption{\small Wage share $\omega$ and employment rate $e$ (top row), net financial charges $\bar r_L \ell -\bar r_M m$ and speculative flow $f$ (middle row), and stock price $S$, trend indicator $\mu$ and effective rate $r$ (bottom row) for the cases with $\bar\eta_\mu = 5, \bar\rho_2 = 3$ (left) and $\bar\eta_\mu = 0.2, \bar\rho_2 = 3$.} \label{fig:coupled_e_f}
\end{figure}

\subsection{Parameter sensitivity}

In the previous section we explore the trajectories of key variables in the model under different combinations of parameters. Broadly speaking, higher interest rates for longer periods lead to unsustainable accumulation of debt and an eventual collapse of the economy. In this section, we explore the sensitivity of the model with respect to some of the newly introduced parameters in a more systematic way. 

Specifically, because the model is stochastic, we simulate it multiple times for each parameter combination and compute the proportion of the runs that exhibit an economic collapse within 150 years, defined as a state in which one of the following two conditions is reached:
\begin{equation}
    \label{eq:crisis definition}
    e\le 0.05 \; \quad \text{or} \quad \ell - m \ge 10\;.
\end{equation}
Typically, in the model of \cite{Keen1995} and its extensions, low employment rate and high private debt tend to occur together in the same explosive equilibrium, as shown in \cite{GrasselliCostaLima2012} and \cite{GrasselliNguyenHuu2015}. In our simulations, we verify them as separate conditions as they represent, in principle, different catastrophic outcomes for the economy. We also classify as a collapse any run in which the numerical simulation of the model fails to produce a result.    

The top-left panel of Figure \ref{fig:sensitivity} confirms that a higher lending rate $\bar r_L$ leads to more frequent crises, as previously observed in the right panels of Figures \ref{fig:decoupled_a_b} to \ref{fig:decoupled_e_f}. In the top-right panel of Figure \ref{fig:sensitivity}, we see that a higher stock price volatility $\bar\sigma$ also leads to more frequent crises, as previously observed in the right panel of Figure \ref{fig:coupled_a_b}. The bottom panels of Figure \ref{fig:sensitivity} confirm the behaviour of the model previously observed in Figures \ref{fig:coupled_c_d} and \ref{fig:coupled_e_f} with respect to the speed of mean reversion $\bar\eta_\mu$ of the trend indicator and the reaction of the interest rate to drops in the trend indicator, as measured by $\bar\rho_2$. Essentially, faster mean reversion leads to fewer crises, as the interest rate spikes are not sustained for long periods, whereas larger spikes in interest rates when the trend indicator drops lead to more frequent crises. 

Figure \ref{fig:heatmaps} reinforces these conclusions by varying two parameters at a time, with the left panel showing that combinations of low values of $\bar\eta_\mu$ and high values of $\bar\rho_2$ lead to more frequent crises, and the right panel showing that high values of $\bar r_L$ and $\bar\sigma$ lead to the same behaviour. 

\begin{figure}[ht!]
\includegraphics[width=.49\linewidth]{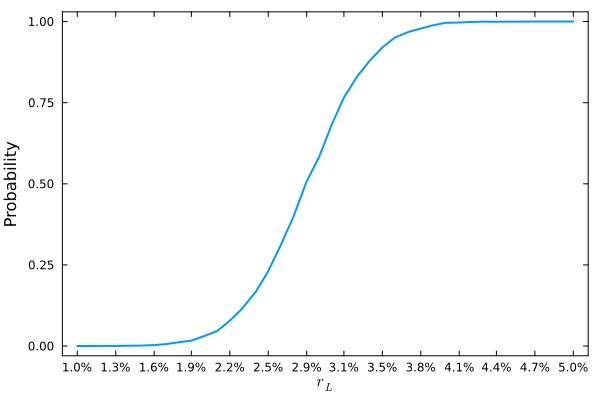}\hfill
\includegraphics[width=.49\linewidth]{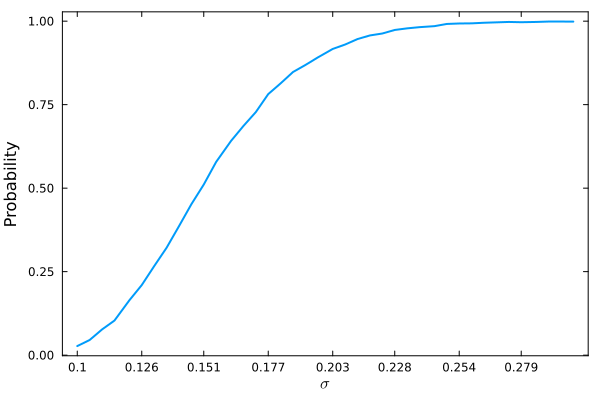}
\includegraphics[width=.49\linewidth]{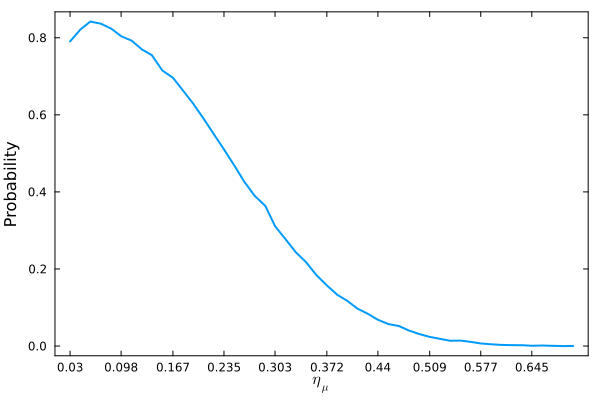}
\includegraphics[width=.49\linewidth]{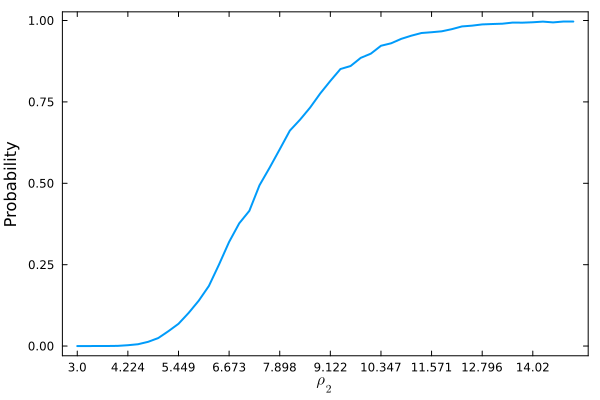}\hfill
\caption{\small Crisis probability with respect to single parameter variations. 
A crisis event is defined as either $e_t\le 0.05$ or $\ell_t - m_t\ge 10$ for some $t\in [0,150]$, or a blow-up in the numerical simulation scheme before $T=150$.
For each parameter value in the horizontal axes above, the probability is estimated via 5000 Monte-Carlo system simulations with all other parameters fixed as in Table \ref{parameters}.} 
\label{fig:sensitivity}
\end{figure}

\begin{figure}[ht!]
\includegraphics[width=.49\linewidth]{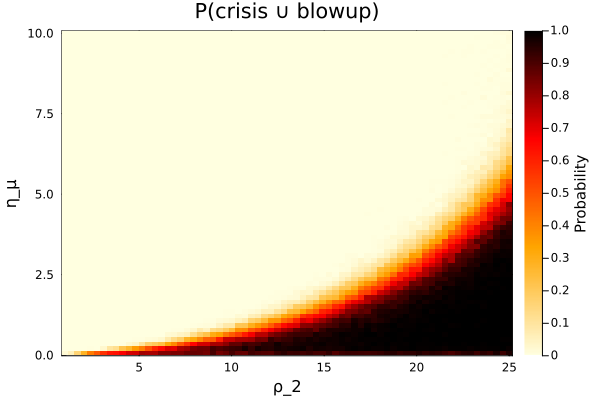}\hfill
\includegraphics[width=.49\linewidth]{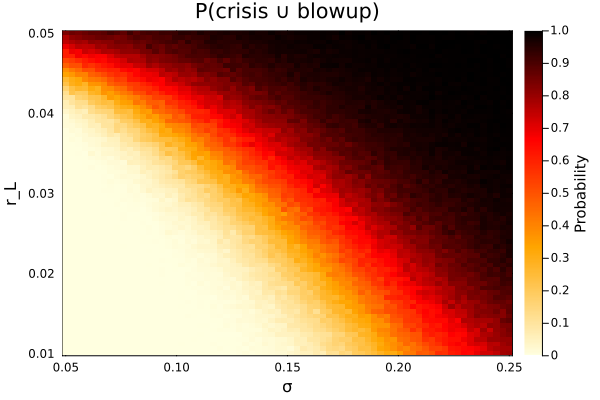}
\caption{\small Heatmaps of crisis probability with respect to two-parameter variations. Crisis event defined as in Figure \ref{fig:sensitivity}. The crisis probability is estimated at each point via 300 Monte-Carlo system simulations, with all other parameters fixed as in Table \ref{parameters}.} 
\label{fig:heatmaps}
\end{figure}

\section{Conclusion}

We have developed a stochastic macro--financial model that unifies debt-driven real dynamics with endogenous financial instability within a stock--flow consistent framework. 
By embedding a compensated jump--diffusion process with state-dependent intensities into a Keen-type macroeconomic core, the model formalizes how speculative credit expansion simultaneously fuels asset price growth and increases crash risk.

The feedback loop between market dynamics and lending spreads generates endogenous boom--bust cycles without relying on exogenous shocks. 
Financial crises emerge as structural features of the system rather than as externally imposed disturbances. 
The mathematical construction ensures global well-posedness, strict positivity of asset prices, and absence of finite-time explosion, providing a rigorous foundation for numerical exploration.

Beyond its immediate analytical results, the framework offers a tractable platform for studying macroprudential regulation, credit policy interventions, and the systemic implications of leverage-dependent crash risk. 
Future work may extend the model toward heterogeneous agents, multi-asset markets, or stochastic monetary policy rules, further bridging dynamical systems theory and macro-financial analysis.

More broadly, the model illustrates how mathematical consistency and economic realism can be combined to study financial instability as an endogenous dynamical phenomenon.
\begin{appendix}

\section{Existence and construction}
\label{sec:existence}

We prove that the real-economy block \eqref{keen_ponzi} does not explode in finite time
when the interest rate $r=r_t$ is time-varying but bounded.
We then construct the asset-price process with state-dependent jump intensities and show strict positivity and non-explosion.
Finally, we obtain existence for the coupled macro-financial system by a pathwise argument.
To ease the notations below, we replace $m_f$ by $m$.

\begin{Proposition}
\label{prop:global_existence}
Let $r=(r_t)_{t\ge 0}$ be a bounded function of time such that
\begin{equation}
\label{eq:r_bounded_appendix}
0 \le r_t \le \bar r_{\max},
\qquad
\forall t\ge 0.
\end{equation}
Let $(\omega_0,e_0,m_{f,0},\ell_0)\in (0,+\infty)\times(0,+\infty)\times\R^2$.
Let $\pi$ be given by \eqref{eq:profit share}, namely
\begin{equation}
\label{eq:pi_appendix}
\pi_t
=
1-\omega_t-\bar\delta\bar\nu+\bar r_M m_t-r_t\ell_t.
\end{equation}
Let $g(\pi)$ be defined by \eqref{eq:growth}, and let
$f_t=\Psi(g(\pi_t)+i(\omega_t))$ with $\Psi$ defined in \eqref{eq:Psi}.
Assume that $\Phi$ is affine as in \eqref{eq:philips},
and that $i(\omega)$ is given by \eqref{inflation}.
Then \eqref{keen_ponzi} admits a unique global solution
\[
(\omega_t,e_t,m_t,\ell_t)_{t\ge 0}
\in
(0,+\infty)\times(0,+\infty)\times\R^2,
\]
and no component can blow up in finite time.
\end{Proposition}

\begin{proof}
We view \eqref{keen_ponzi} as a non-autonomous ODE
\[
\dot x_t = F(t,x_t),
\qquad
x_t=(\omega_t,e_t,m_t,\ell_t),
\]
where the time-dependence only comes from $r_t$.
We recall the standard blow-up alternative for maximal solutions of ODEs:
if $T_{\max}<+\infty$, then
\begin{equation}
\label{eq:blow_up_alternative}
\limsup_{t\uparrow T_{\max}}
\left(
|\omega_t|+|e_t|+|m_t|+|\ell_t|
\right)
=
+\infty.
\end{equation}
We prove that for every $T<+\infty$ the solution is bounded on $[0,T]$.
This implies $T_{\max}=+\infty$ by \eqref{eq:blow_up_alternative}.

\medskip

The functions $\kappa$, $\Delta$ and $\Psi$ are of the ReLU form,
see \eqref{eqn:kappa}, \eqref{eqn:div_function} and \eqref{eq:Psi}.
Hence they are globally Lipschitz and globally bounded.
In particular, define the bounds
\begin{align*}
\label{eq:bounds_kappa_Delta_Psi}
\bar\kappa_{\rm abs}
&:=
\max\big(|\kappa_{\min}|,\bar\kappa_{max}\big),\\
\bar\Delta_{\rm abs}
&:=
\max\big(|\bar\Delta_{min}|,\bar\Delta_{max}\big), \\
\bar\Psi_{\rm abs}
&:=
\max\big(|\bar\Psi_{min}|,\bar\Psi_{max}\big).
\end{align*}
Then, for all $u\in\R$,
\begin{equation}
\label{eq:kappa_Delta_Psi_abs_bounds}
|\kappa(u)|\le \bar\kappa_{\rm abs},
\qquad
|\Delta(u)|\le \bar\Delta_{\rm abs},
\qquad
|\Psi(u)|\le \bar\Psi_{\rm abs}.
\end{equation}
Moreover, $\Phi$ is affine by \eqref{eq:philips}.
The inflation function $i(\omega)$ is affine by \eqref{inflation}.
Since $\pi$ is affine in $(\omega,m,\ell)$ for each fixed $r_t$,
the vector field $x\mapsto F(t,x)$ is locally Lipschitz on
$(0,+\infty)\times(0,+\infty)\times\R^2$,
uniformly in $t$ on bounded intervals, because of \eqref{eq:r_bounded_appendix}.
Therefore, by the Cauchy--Lipschitz theorem, there exists a unique maximal
solution on some interval $[0,T_{\max})$ with $T_{\max}\in(0,+\infty]$.

The first two equations of \eqref{keen_ponzi} can be written as
\begin{equation}
\label{eq:omega_mult}
\dot\omega_t
=
\omega_t
\left[
\Phi(e_t)-\bar\alpha-(1-\bar\gamma)i(\omega_t)
\right],
\end{equation}
\begin{equation}
\label{eq:e_mult}
\dot e_t
=
e_t
\left[
g(\pi_t)-\bar\alpha-\bar\beta
\right].
\end{equation}
Hence, for all $t<T_{\max}$,
\begin{equation}
\label{eq:omega_exp_form}
\omega_t
=
\omega_0
\exp\left(
\int_0^t
\left[
\Phi(e_s)-\bar\alpha-(1-\bar\gamma)i(\omega_s)
\right]ds
\right),
\end{equation}
\begin{equation}
\label{eq:e_exp_form}
e_t
=
e_0
\exp\left(
\int_0^t
\left[
g(\pi_s)-\bar\alpha-\bar\beta
\right]ds
\right).
\end{equation}
Therefore $\omega_t>0$ and $e_t>0$ for all $t<T_{\max}$.
From \eqref{eq:growth} and the bound $\kappa(\pi)\le \bar\kappa_{\rm abs}$, 
we obtain for all $t<T_{\max}$,
\begin{equation}
\label{eq:g_upper}
g(\pi_t)
=
\frac{\kappa(\pi_t)}{\bar\nu}-\bar\delta
\le
\frac{\bar\kappa_{\rm abs}}{\bar\nu}-\bar\delta
=:
g_{\max}.
\end{equation}
Plugging \eqref{eq:g_upper} into \eqref{eq:e_mult} yields
\begin{equation}
\label{eq:e_gronwall}
\dot e_t
\le
\left(g_{\max}-\bar\alpha-\bar\beta\right)e_t.
\end{equation}
By Gronwall's lemma, for every $T<+\infty$ and all
$t\in[0,T]\cap[0,T_{\max})$,
\begin{equation}
\label{eq:e_bound_T}
e_t
\le
e_0\exp\left(\left(g_{\max}-\bar\alpha-\bar\beta\right)t\right)
\le
e_0\exp\left(\left(g_{\max}-\bar\alpha-\bar\beta\right)T\right)
=:E_T.
\end{equation}
From \eqref{inflation} and the positivity of $\omega$, we have
\begin{equation}
\label{eq:i_lower}
i(\omega)=\bar\eta_p(\bar\xi\omega-1)\ge -\bar\eta_p.
\end{equation}
Hence, for all $t<T_{\max}$,
\begin{equation}
\label{eq:omega_ineq_1}
\dot\omega_t
=
\omega_t\left[\Phi(e_t)-\bar\alpha-(1-\bar\gamma)i(\omega_t)\right]
\le
\omega_t\left[\Phi(e_t)-\bar\alpha+(1-\bar\gamma)\bar\eta_p\right].
\end{equation}

Let $T<+\infty$ and use the bound \eqref{eq:e_bound_T}.
Since $\Phi(e)=\bar\Phi_0+\bar\Phi_1 e$ by \eqref{eq:philips},
we obtain on $[0,T]\cap[0,T_{\max})$,
\begin{equation}
\label{eq:Phi_upper_T}
\Phi(e_t)
=
\bar\Phi_0+\bar\Phi_1 e_t
\le
|\bar\Phi_0|+|\bar\Phi_1|\,E_T
=:\Phi_T.
\end{equation}
Combining \eqref{eq:omega_ineq_1} and \eqref{eq:Phi_upper_T}, we get
\begin{equation}
\label{eq:omega_gronwall}
\dot\omega_t
\le
\left(\Phi_T-\bar\alpha+(1-\bar\gamma)\bar\eta_p\right)\omega_t.
\end{equation}
By Gronwall's lemma, for all $t\in[0,T]\cap[0,T_{\max})$,
\begin{equation}
\label{eq:omega_bound_T}
\omega_t
\le
\omega_0
\exp\left(\left(\Phi_T-\bar\alpha+(1-\bar\gamma)\bar\eta_p\right)t\right)
\le
\omega_0
\exp\left(\left(\Phi_T-\bar\alpha+(1-\bar\gamma)\bar\eta_p\right)T\right)
=:\Omega_T.
\end{equation}

Let us turn to variables $m$ and $\ell$.
Fix $T<+\infty$.
On $[0,T]\cap[0,T_{\max})$ we have $0<\omega_t\le \Omega_T$ and
$0<e_t\le E_T$ by \eqref{eq:e_bound_T}--\eqref{eq:omega_bound_T},
and $0\le r_t\le \bar r_{\max}$ by \eqref{eq:r_bounded_appendix}.

From \eqref{eq:pi_appendix} and $0\le r_t\le \bar r_{\max}$, we obtain
\begin{equation}
\label{eq:pi_abs_bound}
|\pi_t|
\le
\left|1-\bar\delta\bar\nu\right|+|\omega_t|
+\bar r_M |m_t|+\bar r_{\max}|\ell_t|
\le
\left|1-\bar\delta\bar\nu\right|+\Omega_T
+\bar r_M |m_t|+\bar r_{\max}|\ell_t|.
\end{equation}

Next, since $\kappa(\cdot)$ is bounded between $\kappa_{\min}$ and
$\bar\kappa_{max}$ by \eqref{eqn:kappa}, we have for all $u\in\R$,
\begin{equation}
\label{eq:g_abs_bound}
|g(u)|
\le
\left|\frac{\bar\kappa_{\rm abs}}{\bar\nu}-\bar\delta \right| := \bar g_{\rm abs}.
\end{equation}
for all $t<T_{\max}$.
% \begin{equation}
% \label{eq:g_abs_bound}
% |g(\pi_t)|
% \le
% \bar g_{\rm abs}
% :=
% \max\left(
% \left|\frac{\kappa_{\min}}{\bar\nu}-\bar\delta\right|,
% \left|\frac{\bar\kappa_{max}}{\bar\nu}-\bar\delta\right|
% \right).
% \end{equation}

Finally, by \eqref{inflation} and the bound $\omega_t\le \Omega_T$,
\begin{equation}
\label{eq:i_abs_bound}
|i(\omega_t)|
=
\left|\bar\eta_p(\bar\xi\omega_t-1)\right|
\le
\bar\eta_p\left(\bar\xi\omega_t+1\right)
\le
\bar\eta_p\left(\bar\xi\Omega_T+1\right)
=: I_T.
\end{equation}

Combining \eqref{eq:g_abs_bound} and \eqref{eq:i_abs_bound}, we set
\begin{equation}
\label{eq:H_T_def}
H_T
:=
\bar g_{\rm abs}+I_T,
\qquad\text{so that}\qquad
|g(\pi_t)+i(\omega_t)|\le H_T
\quad\text{on }[0,T]\cap[0,T_{\max}).
\end{equation}

Also, define
\begin{equation}
\label{eq:R_L_def}
\sup_{t\in[0,T]}|r_t-\bar\kappa_L|
\le
\bar r_{\max}+|\bar\kappa_L| =: R_L
\end{equation}
Using the equations for $\dot m$ and $\dot\ell$ in \eqref{keen_ponzi},
together with \eqref{eq:kappa_Delta_Psi_abs_bounds} and \eqref{eq:R_L_def},
we obtain for all $t\in[0,T]\cap[0,T_{\max})$,
\begin{equation}
\label{eq:mdot_abs_1}
|\dot m_t|
\le
|\pi_t|
+
(1-\bar\zeta)\,|\kappa(\pi_t)-\bar\delta\bar\nu|
+
|r_t-\bar\kappa_L|\,|\ell_t|
+
|\Delta(\pi_t)|
+
|f_t|
+
|g(\pi_t)+i(\omega_t)|\,|m_t|.
\end{equation}
Moreover,
\begin{equation}
\label{eq:kappa_minus_dnu_bound}
|\kappa(\pi_t)-\bar\delta\bar\nu|
\le
|\kappa(\pi_t)|+\bar\delta\bar\nu
\le
\bar\kappa_{\rm abs}+\bar\delta\bar\nu.
\end{equation}
Also $|f_t|=|\Psi(g(\pi_t)+i(\omega_t))|\le \bar\Psi_{\rm abs}$ by
\eqref{eq:kappa_Delta_Psi_abs_bounds}.

Plugging \eqref{eq:pi_abs_bound}, \eqref{eq:kappa_minus_dnu_bound},
\eqref{eq:R_L_def}, \eqref{eq:kappa_Delta_Psi_abs_bounds}, and
\eqref{eq:H_T_def} into \eqref{eq:mdot_abs_1} yields
\begin{equation}
\label{eq:mdot_abs_2}
|\dot m_t|
\le
A_T^{(m)}
+
B_T^{(m)}\left(|m_t|+|\ell_t|\right),
\end{equation}
where the constants can be chosen explicitly as
\begin{equation}
\label{eq:A_B_m_explicit}
A_T^{(m)}
:=
\left|1-\bar\delta\bar\nu\right|+\Omega_T
+
(1-\bar\zeta)\left(\bar\kappa_{\rm abs}+\bar\delta\bar\nu\right)
+
\bar\Delta_{\rm abs}
+
\bar\Psi_{\rm abs},
\end{equation}
\begin{equation}
\label{eq:B_m_explicit}
B_T^{(m)}
:=
\max\left(
\bar r_M+H_T,\;
\bar r_{\max}+R_L
\right)
=
\max\left(
\bar r_M+H_T,\;
2\bar r_{\max}+|\bar\kappa_L|
\right).
\end{equation}

Similarly, from the equation for $\dot\ell$ in \eqref{keen_ponzi},
\begin{equation}
\label{eq:ldot_abs_1}
|\dot\ell_t|
\le
\bar\zeta\,|\kappa(\pi_t)-\bar\delta\bar\nu|
+
|r_t-\bar\kappa_L|\,|\ell_t|
+
|f_t|
+
|g(\pi_t)+i(\omega_t)|\,|\ell_t|.
\end{equation}
Using \eqref{eq:kappa_minus_dnu_bound}, \eqref{eq:R_L_def},
\eqref{eq:kappa_Delta_Psi_abs_bounds}, and \eqref{eq:H_T_def},
we obtain
\begin{equation}
\label{eq:ldot_abs_2}
|\dot\ell_t|
\le
A_T^{(\ell)}
+
B_T^{(\ell)}\left(|m_t|+|\ell_t|\right),
\end{equation}
with explicit constants
\begin{equation}
\label{eq:A_ell_explicit}
A_T^{(\ell)}
:=
\bar\zeta\left(\bar\kappa_{\rm abs}+\bar\delta\bar\nu\right)
+
\bar\Psi_{\rm abs},
\end{equation}
\begin{equation}
\label{eq:B_ell_explicit}
B_T^{(\ell)}
:=
R_L+H_T
\le
\bar r_{\max}+|\bar\kappa_L|+H_T.
\end{equation}
Let
\begin{equation}
\label{eq:A_B_for_u}
A_T
:=
A_T^{(m)}+A_T^{(\ell)},
\qquad
B_T
:=
B_T^{(m)}+B_T^{(\ell)}.
\end{equation}
Then \eqref{eq:mdot_abs_2} and \eqref{eq:ldot_abs_2} imply that for all $t$
\begin{equation}
\label{eq:u_derivative}
\frac{d}{dt}\left(|m_t|+|\ell_t|\right)
\le
|\dot m_t|+|\dot\ell_t|
\le
A_T+B_T\left(|m_t|+|\ell_t|\right).
\end{equation}
By Gronwall's lemma, for all $t\in[0,T]\cap[0,T_{\max})$,
\begin{equation}
\label{eq:u_bound}
|m_t|+|\ell_t|
\le
\left(|m_0|+|\ell_0|\right)e^{B_T t}
+
\frac{A_T}{B_T}\left(e^{B_T t}-1\right),
\end{equation}
with the obvious modification when $B_T=0$.
In particular, $(m_t,\ell_t)$ are bounded on $[0,T]$.

Let us wrap up \emph{a priori} estimates.
Let $T<+\infty$.
We showed that $\omega$, $e$, $m$, and $\ell$ are bounded on $[0,T]$.
Since $T$ is arbitrary, \eqref{eq:blow_up_alternative} cannot occur.
Hence $T_{\max}=+\infty$.
This proves global existence and uniqueness.
\end{proof}

\begin{Proposition}
\label{prop:S_exist_twosided}
Let $f=(f_t)_{t\ge 0}$ be an $\F$-adapted process with c\`adl\`ag paths and define
\[
f^+_t=\max(f_t,0),
\qquad
f^-_t=\max(-f_t,0).
\]
Let $\bar\lambda^+>0$ and $\bar\lambda^->0$, and set the (predictable) intensities
\[
\lambda^+_t := \bar\lambda^+ f^+_t,
\qquad
\lambda^-_t := \bar\lambda^- f^-_t.
\]
Assume that for every $T<\infty$,
\begin{equation}
\label{eq:intensity_integrable_twosided}
\int_0^T \lambda^+_s\,ds<\infty
\quad\text{and}\quad
\int_0^T \lambda^-_s\,ds<\infty
\qquad\text{a.s.}
\end{equation}
(For instance, \eqref{eq:intensity_integrable_twosided} holds if $f$ is a.s.\ bounded on $[0,T]$.)

Let $(\Omega,\F,(\F_t)_{t\ge 0},\P)$ support a Brownian motion $W$ and two independent
unit-rate Poisson processes $\widehat N^+$ and $\widehat N^-$. Define the time changes
\[
e^\pm(t):=\int_0^t \lambda^\pm_s\,ds,
\qquad t\ge 0,
\]
and the Cox processes
\[
N^\pm_t := \widehat N^\pm_{e^\pm(t)},
\qquad t\ge 0,
\]
with (stochastic) intensities $\lambda^\pm_t$.
Assume $\bar\sigma>0$, $0<\bar J^+<1$ and $\bar J^->0$, and consider the compensated jump--diffusion
\begin{equation}
\label{eq:S_sde_twosided}
dS_t = \bar r_L S_{t-}\,dt + \bar\sigma S_{t-}\,dW_t
- \bar J^+ S_{t-}\big(dN^+_t-\lambda^+_t dt\big)
+ \bar J^- S_{t-}\big(dN^-_t-\lambda^-_t dt\big),
\end{equation}
with $S_0>0$.
Then \eqref{eq:S_sde_twosided} admits a unique strong c\`adl\`ag solution $(S_t)_{t\ge 0}$, given explicitly by
\begin{align}
\label{eq:S_explicit_twosided}
S_t
=
S_0\exp\!\Big(\big(\bar r_L-\tfrac12\bar\sigma^2\big)t+\bar\sigma W_t
+\int_0^t\big(\bar J^+\lambda^+_s-\bar J^-\lambda^-_s\big)\,ds\Big)\,
(1-\bar J^+)^{N^+_t}(1+\bar J^-)^{N^-_t}.
\end{align}
In particular, $S_t>0$ a.s.\ for all $t\ge 0$ and $S$ cannot explode in finite time.
\end{Proposition}

\begin{proof}
Fix $T<\infty$ and work on $[0,T]$. Under \eqref{eq:intensity_integrable_twosided}, the time changes
$e^\pm(t)=\int_0^t\lambda^\pm_s ds$ are finite for all $t\le T$, hence $N^\pm_t=\widehat N^\pm_{e^\pm(t)}$
are well-defined counting processes with finitely many jumps on $[0,T]$ almost surely.

Rewrite \eqref{eq:S_sde_twosided} in the non-compensated form:
\[
dS_t
=
S_{t-}\Big(\bar r_L+\bar J^+\lambda^+_t-\bar J^-\lambda^-_t\Big)\,dt
+\bar\sigma S_{t-}\,dW_t
-\bar J^+ S_{t-}\,dN^+_t
+\bar J^- S_{t-}\,dN^-_t .
\]
This is a linear SDE with multiplicative Brownian part and multiplicative jumps.
Therefore the Dol\'eans--Dade exponential yields the explicit representation
\eqref{eq:S_explicit_twosided}. (Equivalently, apply It\^o's formula with jumps to $\log S_t$.)

Positivity follows from $S_0>0$, $(1-\bar J^+)^{N^+_t}>0$ because $0<\bar J^+<1$, and
$(1+\bar J^-)^{N^-_t}>0$ because $\bar J^->0$. The explicit form \eqref{eq:S_explicit_twosided}
precludes finite-time blow-up on $[0,T]$, and since $T$ is arbitrary, on $\R_+$.
Uniqueness holds because the equation is linear in $S$ with given drivers $(W,N^+,N^-)$.
\end{proof}

\begin{Proposition}
\label{prop:mu_exist}
Assume the setting of Proposition~\ref{prop:S_exist_twosided}. Define the predictable drift term
\begin{align}
\label{eq:a_drift_def}
a_t
:=
\bar r_L-\tfrac12\bar\sigma^2
+\big(\log(1-\bar J^+)+\bar J^+\big)\lambda^+_t
+\big(\log(1+\bar J^-)-\bar J^-\big)\lambda^-_t,
\end{align}
and the c\`adl\`ag local martingale increment
\begin{align}
\label{eq:dM_def}
dM_t
:=
\bar\sigma\,dW_t
+\log(1-\bar J^+)\big(dN^+_t-\lambda^+_t dt\big)
+\log(1+\bar J^-)\big(dN^-_t-\lambda^-_t dt\big).
\end{align}
Let $\bar\eta_\mu>0$ and $\mu_0\in\R$. Consider the linear SDE
\begin{equation}
\label{eq:mu_sde_appendix}
d\mu_t = \bar\eta_\mu\big(a_t-\mu_t\big)\,dt + dM_t,
\qquad \mu_0\in\R.
\end{equation}
Then \eqref{eq:mu_sde_appendix} admits a unique strong c\`adl\`ag solution $(\mu_t)_{t\ge 0}$ given by
\begin{equation}
\label{eq:mu_explicit}
\mu_t
=
e^{-\bar\eta_\mu t}\mu_0
+\int_0^t \bar\eta_\mu e^{-\bar\eta_\mu (t-s)} a_s\,ds
+\int_0^t e^{-\bar\eta_\mu (t-s)}\,dM_s,
\qquad t\ge 0.
\end{equation}
Moreover, $\mu$ cannot explode in finite time.
\end{Proposition}

\begin{proof}
Fix $T<\infty$. Under Proposition~\ref{prop:S_exist_twosided}, the processes $\lambda^\pm$ are predictable
and integrable on $[0,T]$, hence $a_t$ is predictable and finite for a.e.\ $t\le T$.
The process $M$ defined by \eqref{eq:dM_def} is a c\`adl\`ag local martingale on $[0,T]$ with finitely many jumps almost surely.

Equation \eqref{eq:mu_sde_appendix} is linear. Multiply both sides by $e^{\bar\eta_\mu t}$ and use
integration by parts:
\[
d\big(e^{\bar\eta_\mu t}\mu_t\big)
=
\bar\eta_\mu e^{\bar\eta_\mu t} a_t\,dt
+ e^{\bar\eta_\mu t}\,dM_t.
\]
Integrating from $0$ to $t$ yields \eqref{eq:mu_explicit}. This provides a strong solution adapted to the filtration generated by $(W,N^+,N^-)$. Uniqueness follows from linearity: if two solutions exist, their difference solves $d(\delta\mu_t)=-\bar\eta_\mu\,\delta\mu_t\,dt$ with $\delta\mu_0=0$, hence $\delta\mu\equiv 0$.

Finally, on each finite horizon $[0,T]$, the first two terms in \eqref{eq:mu_explicit} are finite, and
the stochastic convolution $\int_0^t e^{-\bar\eta_\mu (t-s)}\,dM_s$ is well-defined and finite because $M$ has finite quadratic variation and finitely many jumps on $[0,T]$. Therefore $\mu$ cannot explode in finite time.
\end{proof}

\begin{Corollary}
\label{cor:coupled_existence}
Assume the hypotheses of Proposition~\ref{prop:global_existence} for the real-economy block, and assume moreover that the rate map $\rho:\R\to[0,\bar r_{\max}]$ is bounded and locally Lipschitz.

Consider the coupled system defined as follows:
\begin{itemize}
\item the real-economy variables $(\omega,e,m,\ell)$ solve \eqref{keen_ponzi} with interest rate $r_t=\rho(\mu_t)$,
\item the normalized speculative flow $f_t$ is generated by the real-economy block (as in Proposition~\ref{prop:global_existence}),
\item the asset price $S$ solves \eqref{eq:S_sde_twosided} with intensities $\lambda^\pm_t=\bar\lambda^\pm f^\pm_t$,
\item the trend indicator $\mu$ solves \eqref{eq:mu_sde_appendix}.
\end{itemize}
Then, for any admissible initial condition $(\omega_0,e_0,m_0,\ell_0,S_0,\mu_0)$ with $\omega_0>0$, $e_0>0$ and $S_0>0$, the coupled system admits a unique global strong c\`adl\`ag solution, unique up to indistinguishability.
\end{Corollary}

\begin{proof}
\textbf{Step 1 (local existence/uniqueness).}
On any bounded domain in the state variables, the drift coefficients of the real-economy block are locally Lipschitz because $\Phi$ and $i$ are affine and the ReLU-type functions $\kappa,\Delta,\Psi$ are globally Lipschitz (hence locally Lipschitz). The feedback $r_t=\rho(\mu_t)$ is locally Lipschitz in $\mu$ by assumption. The intensities $\lambda^\pm_t=\bar\lambda^\pm f^\pm_t$ are locally Lipschitz functions of $f_t$ and are nonnegative. Therefore the coupled dynamics define a locally Lipschitz stochastic system driven by $(W,\widehat N^+,\widehat N^-)$, hence there exists a unique maximal strong solution up to an explosion time $\tau\in(0,\infty]$.

\textbf{Step 2 (a priori bounds and non-explosion).}
Fix $T<\infty$. Since $\rho$ takes values in $[0,\bar r_{\max}]$, we have $0\le r_t\le \bar r_{\max}$ for all $t\ge 0$ almost surely. Proposition~\ref{prop:global_existence} then applies pathwise on $[0,T\wedge\tau)$ and yields that $(\omega_t,e_t,m_t,\ell_t)$ remain finite and cannot blow up on $[0,T\wedge\tau)$. In particular, the generated flow $f_t=\Psi(g(\pi_t)+i(\omega_t))$ is bounded on $[0,T\wedge\tau)$ because $\Psi$ is bounded, hence $\lambda^\pm_t=\bar\lambda^\pm f^\pm_t$ are bounded on $[0,T\wedge\tau)$ and satisfy the integrability condition \eqref{eq:intensity_integrable_twosided}.

Proposition~\ref{prop:S_exist_twosided} then yields that $S_t$ is strictly positive and finite on $[0,T\wedge\tau)$, and Proposition~\ref{prop:mu_exist} yields that $\mu_t$ is finite on $[0,T\wedge\tau)$. Consequently, all components of the coupled state remain finite on $[0,T\wedge\tau)$, so explosion cannot occur before $T$. Since $T$ is arbitrary, $\tau=+\infty$ almost surely.

\textbf{Step 3 (global uniqueness).}
Uniqueness on each $[0,T]$ follows from Step~1, hence the global solution is unique up to indistinguishability.
\end{proof}

\section*{Declaration on the Use of Generative AI}

The authors acknowledge the use of generative AI tools for language editing and structural refinement of portions of the manuscript. All mathematical derivations, model constructions, simulations, and economic interpretations were developed, verified, and validated by the authors. The authors take full responsibility for the accuracy, originality, and integrity of the content.

\end{appendix}

%%%%%%%%%%%%%%%%%%%%%%%%%%%%%%%%%%%%%%%%%%%%%%%%%%%%%%%%%%%%%%%%%%%%%%%%%%%
%%%%%%%%%%%%%%%%%%%%%%%%%%%%%%%%%%%%%%%%%%%%%%%%%%%%%%%%%%%%%%%%%%%%%%%%%%%
%%%%%%%%%%%%%%%%%%%%%%%%%%%%%%%%%%%%%%%%%%%%%%%%%%%%%%%%%%%%%%%%%%%%%%%%%%%
\bibliographystyle{plainnat}
\bibliography{finance}

\end{document}